\documentclass[aip,amsmath,amssymb,reprint]{revtex4-1}
\usepackage[T1]{fontenc}
\usepackage[utf8]{inputenc}
\usepackage{graphicx}
\usepackage{dcolumn}
\usepackage{amsthm}
\usepackage{empheq}
\usepackage{dsfont}
\usepackage{stmaryrd}
\usepackage[version=4]{mhchem}
\usepackage{mathtools}
\usepackage{color}
\usepackage{xspace}

\makeatletter
\renewcommand*\@xhline{\ifx\reserved@a\hline
               \vskip\doublerulesep
               \vskip-\arrayrulewidth
             \fi
      \ifnum0=`{\fi}}
\makeatother

\usepackage{siunitx}
\usepackage{subcaption}
\usepackage{booktabs}
\usepackage{hyperref}
\usepackage{cleveref}
\usepackage{pgfplots}
\usepackage{microtype}
\usepackage{xcolor}

\SetSymbolFont{stmry}{bold}{U}{stmry}{m}{n}

\renewcommand*{\doi}[1]{\href{http://dx.doi.org/#1}{doi:\nolinkurl{#1}}}

\graphicspath{{./figures/}}
\usetikzlibrary{matrix}
\usetikzlibrary{patterns}
\usetikzlibrary{shapes}
\pgfplotsset{compat=newest}
\usepgfplotslibrary{groupplots}
\usepgfplotslibrary{external}
\usepgfplotslibrary{colorbrewer}
\hypersetup{
  pdfauthor = {Eirik Holm Fyhn, Karl Yngve Lervåg, Åsmund Ervik and Øivind
  Wilhelmsen},
  pdftitle = {A consistent reduction of the two-layer shallow-water equations to
  an accurate one-layer spreading model},
  pdfcreator = {LaTeX with hyperref package},
  pdfproducer = {dvips, ps2pdf},
  bookmarksnumbered,
  pdfstartview = {FitH},
  colorlinks = false,
}

\crefname{figure}{figure}{figures}

\newtheorem{theorem}{Theorem}
\newtheorem{lemma}{Lemma}

\newtheorem{definition}{Definition}

\newcommand*{\eg}{e.g.\xspace}
\newcommand*{\ie}{i.e.\xspace}
\newcommand*{\diff}{\ensuremath{\mathop{}\!\mathrm{d}}}
\newcommand*{\defeq}{\vcentcolon=}
\newcommand*{\pd}[1]{\ensuremath{\frac{\partial}{\partial #1}}}
\newcommand*{\pdt}[1]{\ensuremath{\frac{\partial #1}{\partial t}}}
\newcommand*{\pdfun}[2]{\ensuremath{\frac{\partial #1}{\partial #2}}}
\newcommand*{\Fr}{\ensuremath{\mathrm{Fr}}}
\newcommand*{\FrLe}{\ensuremath{\Fr_{\text{\tiny LE}}}}
\newcommand*{\Le}[1]{{#1}_{\text{\tiny LE}}}
\newcommand{\vect}[1]{\boldsymbol{#1}}
\renewcommand{\div}{\ensuremath{\nabla\cdot}}

\newcommand*{\jump}[1]{\ensuremath{\left\llbracket{#1}\right\rrbracket}}
\newcommand*{\av}[1]{\ensuremath{\left\langle{#1}\right\rangle}}
\newcommand*{\spaceav}[1]{\ensuremath{\overline{#1}}}
\newcommand*{\abs}[1]{\ensuremath{\left\lvert{#1}\right\rvert}}
\newcommand*{\norm}[1]{\ensuremath{\left\lVert{#1}\right\rVert}_1}
\newcommand*{\n}{\ensuremath{\vect{\hat n}}}
\newcommand*{\ndot}[1]{\ensuremath{\vect{\hat n}\cdot\vect{#1}}}
\newcommand*{\nd}[0]{\ensuremath{\vect{\hat n}\cdot}}
\newcommand*{\escaled}{\ensuremath{\widetilde{E_k}}}

\begin{document}


\title{A consistent reduction of the two-layer shallow-water equations to an accurate one-layer spreading model}

\author{Eirik Holm Fyhn}
\affiliation{NTNU, Department of Physics, Høgskoleringen 5, NO-7491 Trondheim, Norway}
\author{Karl Yngve Lervåg}
\author{Åsmund Ervik}
\affiliation{SINTEF Energy Research, P.O. Box 4671 Sluppen, NO-7465 Trondheim, Norway}
\author{Øivind Wilhelmsen}
\affiliation{SINTEF Energy Research, P.O. Box 4671 Sluppen, NO-7465 Trondheim, Norway}
\affiliation{NTNU, Department of Energy and Process Engineering, NO-7465 Trondheim, Norway}

\date{\today}

\begin{abstract}
  The gravity-driven spreading of one fluid in contact with another fluid is of key importance to a range of topics.
  These phenomena are commonly described by the two-layer shallow-water equations (SWE).
  When one layer is significantly deeper than the other, it is common to approximate the system with the much simpler one-layer SWE.
  It has been assumed that this approximation is invalid near shocks, and one has applied additional front conditions to correct the shock speed.
  In this paper, we prove mathematically that an effective one-layer model can be derived from the two-layer equations that correctly captures the behaviour of shocks and contact discontinuities without additional closure relations.
  The result shows that simplification to an effective one-layer model is  justified mathematically and can be made without additional knowledge of the shock behaviour.
  The shock speed in the proposed model is consistent with empirical models and identical to front conditions that have been found theoretically by \eg~von Kármán and by Benjamin.
  This suggests that the breakdown of the SWE in the vicinity of shocks is less severe than previously thought.
  We further investigate the applicability of the SW framework to shocks by studying one-dimensional lock-exchange/-release.
  We derive expressions for the Froude number that are in good agreement with the widely employed expression by Benjamin.
  The equations are solved numerically to illustrate how quickly the proposed model converges to solutions of the full two-layer SWE.
  We also compare numerical results from the model with results from experiments, and find good agreement.
\end{abstract}

\maketitle

\section{Introduction}
\label{sec:Introduction}
The spreading of two layers of fluids with different density is of considerable importance.
It has been an active field of study since at least 1774, when \citet{franklin1774} investigated how oil spreads on water and how this can be used to still waves.
Applications where this phenomenon plays an important role include spills of oil~\citep{hoult1972, fay1971, chebbi2001} and liquefied gaseous fuels,\citep{fay2007, fay2003, brandeis1983_part2} stratified flow inside pipes,\citep{stanislav1986} gravity currents particularly in geophysical systems,\citep{adduce2012, shin2004, moodie2002,meriaux2016} monomolecular layers for evaporation control,\citep{stickland1972} and coalescence in three-phase fluid systems.\citep{mar1968}
These applications include non-miscible fluids such as oil and water, or systems with miscible fluids at large Richardson number, \ie~where buoyancy dominates mixing effects and ensures separation into layers.

A fundamental property of spreading phenomena is the rate of spreading, or the speed of the leading edge of the spreading fluid.
This is typically characterized by the dimensionless Froude number,\citep{white2003,vaughan2005}
\begin{equation}
  \Fr = \frac{u}{\sqrt{g'h}},
  \label{eq:Fr}
\end{equation}
where $u$ is the velocity, $h$ is the height of the layer that is spreading, and $g'$ is the effective gravitational acceleration.
In two layer spreading, the effective gravitational acceleration is $g' = (1-\rho_1/\rho_2)g$, where $\rho_1$ and $\rho_2$ are the two fluid densities and $\rho_1 < \rho_2$.

An early result for the Froude number of gravity currents was presented by \citeauthor{karman1940}.\cite{karman1940}
They found that for the edge of a spreading gravity current at semi-infinite depth, $\FrLe = \sqrt{2}$, where the subscript is short for ``leading edge''.
\citet{benjamin1968} later developed a model for $\FrLe$ for spreading of gravity currents with constant height,
\begin{equation}
  \FrLe^2 = \frac{(1-\alpha)(2-\alpha)}{(1+\alpha)},
  \label{eq:benjamin}
\end{equation}
where $\alpha = h_2/(h_1+h_2)$.
Here $h_1$ and $h_2$ are the heights of the top and bottom layers, respectively.
This model approaches the result by \citeauthor{karman1940} when the bottom layer becomes thin, $h_2 \ll h_1$.
More recently, \citet{ungarish2017} extended the result of \citeauthor{benjamin1968} to the spreading of gravity currents into a lighter fluid with an open surface.
This result also gives $\FrLe = \sqrt{2}$ when the spreading fluid becoms relatively much thinner than the ambient fluid.

The next step beyond characterizing spreading rates is to develop a model that predicts the phenomenon in more detail.
An early model was presented by \citeauthor{fay1969},\cite{fay1969} who studied the spreading of oil on water.
They divided the spreading into three phases; one where inertial forces dominate, one where viscous forces dominate, and one where the surface tension dominates.
In the inertial phase, the speed of the front can be written as
\begin{equation}
  \Le{u} = \beta \sqrt{\frac{g' V}{A}},
  \label{eq:fay}
\end{equation}
where $\beta$ is an empirical constant and $V$ and $A$ are the volume and area, respectively.
Then $V/A$ is the average height of the spreading oil.
In this model, $\beta$ represents an effective Froude number where the height at the leading edge is approximated by the average height.
The value of $\beta$ has been discussed in the literature and is commonly set to $\beta = 1.31$ in the one-dimensional case and $\beta = 1.41$ in the axisymmetric case.\citep{fannelop1972,fay1971,hoult1972,fay2007}

A more general approach than the Fay model is the two-layer shallow-water equations (2LSWE), which are derived from the Euler equations by assuming a negligible vertical velocity.\citep{ovsyannikov1979,vreugdenhil1979}
These equations model the flow of two layers of shallow liquids and may be used to simulate for instance gravity currents.\citep{audusse11}
However, internal breaking of waves or large differences in velocities of the two layers can break the hyperbolicity of the equations.
Even if the initial conditions are hyperbolic, the system can evolve into a non-hyperbolic state.\citep{milewski2004}
A breakdown of hyperbolicity causes problems such as ill-posedness and Kelvin-Helmholtz like instabilities.\citep{lannes2015,stewart2013,lam2016}
Non-hyperbolic equations are generally more difficult to analyse and computationally much more expensive to solve than hyperbolic equations.\citep{bouchut2008}
Attempts to amend the non-hyperbolicity of the systems include adding numerical (non-physical) friction forces,\citep{castro2011} operator-splitting approaches,\citep{bouchut2010} and introduction of an artificial compressibility.\citep{chiapolino2018}

Due to their comparative simplicity, the one-layer shallow-water equations (1LSWE) have often been used to model two-layer phenomena like liquid-on-liquid spreading and gravity currents where one assumes that the layers are in a buoyant equilibrium.
In this case, a forced constant Froude-number boundary condition at the leading edge of a spreading liquid is used to account for the effect of the missing layer. \citep{fannelop1972,hoult1972,hatcher2014}
The additional boundary condition at the leading edge has also been used in combination with the 2LSWE.\citep{rottman1983,ungarish2013}
In particular, \citet{rottman1983} argued that a front condition that includes the Froude number is necessary because viscous dissipation and vertical acceleration are too significant to be neglected at the front.

The 1LSWE are always hyperbolic and therefore have fewer challenges than the 2LSWE.
However, there are situations where even the 1LSWE are not strictly hyperbolic, meaning that the two eigenvalues of the Jacobian coincide.
This situation is found when considering the wet--dry transition, such as the dam break on a dry bottom, or for certain bottom topographies.
In particular the case of a gravity current flowing upslope, as in a shallow water wave encountering a beach, is of importance and has seen new developments in recent years.\citep{lombardi2015,bjornestad2017,zemach2019}
There is an exhaustive literature on the subject of hyperbolicity of the 1LSWE, \citep{fracarollo1995,zhou2001,lefloch2007,liang2009,lefloch2011,murillo2016} including the topic of well-balanced formulation, the more general E-balanced schemes, and the identification of resonant versus non-resonant regimes of flow.
These points are mainly of interest for the numerical solution of the equations in specific regimes.
As the present paper is focused more on the theoretical developments, a detailed discussion of hyperbolicity is beyond the scope of the present work.

The main results of the present paper are the following.
First, we show that the need to impose boundary conditions or empirical closures for the spreading rate when using the 1LSWE instead of the 2LSWE follows from the different shock behaviour of the two formulations.
Second, we demonstrate that weak solutions of the 2LSWE converge to weak solutions of a \emph{locally conservative} form of the one-layer equations.
This formulation is different from the standard 1LSWE, and removes the need for front conditions.

This is a strong result as it implies that in many situations, such as when considering liquid spills on water or ocean layers in deep water, one may use the much simpler locally conservative 1LSWE even for two-layer spreading phenomena, without the need for additional boundary conditions or closures.
An example is presented in \cref{fig:intro-example}, which illustrates how solutions to different forms of the 1LSWE compare to the solution of the 2LSWE for a dam-break problem. The figure shows a clear difference between the \emph{locally} and \emph{globally} conservative 1LSWE.

\begin{figure}
  \centering
  \tikzsetnextfilename{intro-example}
  \begin{tikzpicture}
    \begin{axis}[
        height = 0.7\columnwidth,
        width = 0.9\columnwidth,
        xlabel = {Length},
        ylabel = {Height},
        axis lines = left,
        ticks = none,
        xmax = 2.1,
        ymin =-0.1,
        ymax = 1.1,
        line width = 1pt,
        legend pos = north east,
        legend style = {xshift=1cm, draw=none, fill=black!5},
        legend cell align = {left},
        table/x = x,
        domain=-2:2,
        samples=500,
        han1/.style = {
          color=Set1-C,
        },
        han2/.style = {
          dashed,
          color=Set1-E,
        },
        hd2/.style = {
          densely dashed,
          color=Set1-B,
        },
      ]

      \addplot[han1] table [y=han] {data/intro-example.txt};
      \addlegendentry{1LSWE Local}

      \addplot[han2]
        {x < -1 ? 1 : (x < 2 ? (2 - x)^2 / 9 : 0)};
      \addlegendentry{1LSWE Global}

      \addplot[hd2] table [y=hd2] {data/intro-example.txt};
      \addlegendentry{2LSWE}
    \end{axis}
  \end{tikzpicture}
  \caption{An example of how solutions from different formulations of the one-layer shallow-water equations (1LSWE Local and 1LSWE Global) compares to those from the two-layer shallow-water equations (2LSWE) for a dam-break problem.}
  \label{fig:intro-example}
\end{figure}
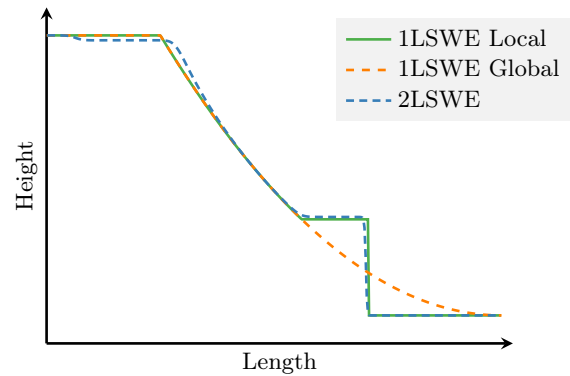

We further demonstrate that the constant Froude number at the front of an expanding fluid can be derived directly from the 2LSWE.
The Froude numbers obtained from the analysis in this paper are in excellent agreement previous results from the literature.
This indicates that the breakdown of the shallow-water equations in vicinity of shocks is less severe than previously suggested.

The paper is structured as following.
In \cref{sec:govEqs}, we introduce the two-layer shallow-water equations (2LSWE), the one-layer shallow-water equations (1LSWE) and the Rankine-Hugoniot condition for the shock.
In \cref{sec:BenjaminAlternative} we derive expressions for the Froude number from the full two-layer shallow-water equations.
The key result of the paper is presented in \cref{sec:reducingToOneLayer}, where we show the 2LSWE can be approximated by a one-layer model when the upper layer is much thicker than the bottom layer, as well as in the opposite situation.
In \cref{sec:cases} we define some  numerical experiments that are used in \cref{sec:results} to study how solutions of the 2LSWE approach the one-layer approximations.
We show that the results from the simplified model are in good agreement with experimental results.
Concluding remarks are provided in \cref{sec:conclusion}.

\section{Theory of the shallow-water equations}
\label{sec:govEqs}
Consider a two-layer system where a fluid of lower density spreads on top of another fluid, as illustrated in \cref{fig:sketch}.
Assuming that the layers are shallow, the solution of the two-layer shallow-water equations (2LSWE) gives the evolution of height and horizontal velocity of both fluids as a function of position and time.

In the following, we first describe the well-known one-layer shallow-water equations (1LSWE).
A straightforward generalization to the 2LSWE is presented next, where we discuss two approaches for reformulating the 2LSWE in a manner that makes them suitable for reduction to an effective one-layer model.
We then show how the Rankine-Hugoniot conditions can be used to predict the shock speed.
Subsequently we employ the vanishing-viscosity regularization and travelling wave solutions to obtain physically acceptable solutions of the partial differential equations (PDEs).
At the end of the section, we present a necessary energy requirement for the 2LSWE that is used to select correct physical solutions in \cref{sec:reducingToOneLayer}.
\begin{figure}
  \centering
  \includegraphics[width=0.9\columnwidth]{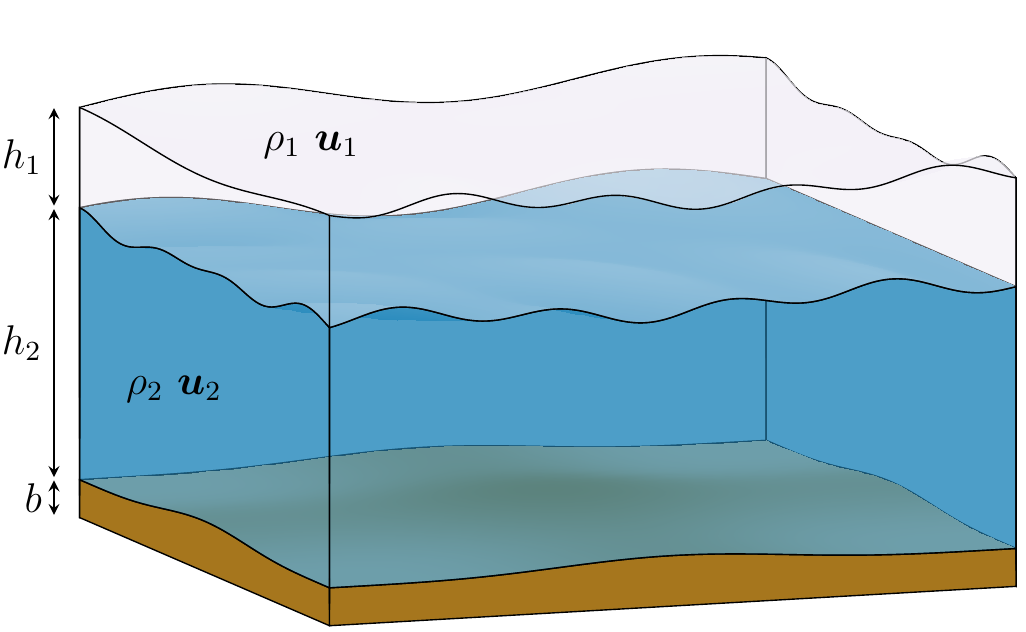}
  \caption{A sketch of a general two-layer shallow-water geometry.}
  \label{fig:sketch}
\end{figure}

\subsection{The one-layer shallow-water equations}
\label{sec:1LSWE}
The 1LSWE are typically presented in a globally conservative form where total momentum is conserved,\citep{leveque02}
\begin{subequations}
  \label{eq:1LSWEGlob}
  \begin{gather}
    \label{eq:1LSWEGlobCont}
    \pd t \rho h + \div (\rho h\vect u) = G_h, \\
    \begin{aligned}[b]
      \pd t (\rho h \vect u) &+ \div (\rho h \vect u \otimes \vect u) \\
      &\qquad+ \nabla\left(\frac 1 2 g \rho h^2\right)
      = \vect G_{hu} - g\rho h\nabla b.
    \end{aligned}
  \end{gather}
\end{subequations}
where $\rho$ is the density, $h$ is the height, $\vect u$ is the \emph{vertically averaged} horizontal velocity, $\otimes$ denotes the tensor product, $b$ is the bottom topography, $G_h$ and $\vect G_{hu}$ are source functions that may represent external phenomena, such as evaporation, Coriolis forces, wind shear stress, or interfacial shear forces.
The bottom topography is assumed to be continuous throughout.
The density $\rho$ is assumed constant in space, although it may vary in time.

One may also consider what will be referred to as the locally conservative 1LSWE, that is
\begin{subequations}
  \label{eq:1LSWELoc}
  \begin{gather}
    \label{eq:1LSWELocCont}
    \pd t \rho h + \div (\rho h \vect u) = G_h, \\
    \pd t \vect u + (\vect u \cdot \nabla)\vect u
    + g \nabla(h + b)
      = \frac 1{\rho h}(\vect G_{hu} - \vect u G_h).
  \end{gather}
\end{subequations}
Here the continuity equation~\eqref{eq:1LSWELocCont} is unchanged.
The various forms of the one-layer and two-layer equations all use the same form of the continuity equation.

One particularly striking difference between \cref{eq:1LSWELoc} and \cref{eq:1LSWEGlob} is the admissibility of shocks when the height drops to 0.
This will be further discussed in \cref{sec:damBreak}, but the upshot is that such a shock is impossible in \cref{eq:1LSWEGlob}, while in \cref{eq:1LSWELoc} it is possible with a Froude number $\FrLe = \sqrt{2}$.
This is exactly the result by \citet{karman1940} for two layer spreading with such shocks.
In fact, in \cref{sec:reducingToOneLayer}, we show that the locally conservative form correctly captures the two-layer behaviour in certain limits.
This result is consistent with previous results which show that numerical approaches will fail to solve the conservation of global momentum.\citep{bouchut2010}

\subsection{The two-layer shallow-water equations}
\label{sec:2LSWE}
The 2LSWE may be written in a general, layerwise form with arbitrary source terms as
\begin{widetext}
  \begin{subequations}
    \label{eq:2LSWEGeneral}
    \begin{align}
      \pd t \rho_1h_1 + \div (\rho_1h_1\vect u_1) &= G_{h_1},
      \label{eq:2LSWEGeneral_h1} \\
      \pd t \rho_2h_2 + \div (\rho_2h_2\vect u_2) &= G_{h_2},
      \label{eq:2LSWEGeneral_h2} \\
      \pd t (\rho_1h_1\vect u_1)
      + \div (\rho_1h_1\vect u_1 \otimes \vect u_1)
      + \nabla \left(\frac 1 2 g\rho_1 h_1^2\right) &= \vect G_{h_1u_1} - g\rho_1h_1\nabla (h_2 + b),
      \label{eq:2LSWEGeneral_mom1} \\
      \pd t (\rho_2h_2\vect u_2)
      + \div (\rho_2h_2\vect u_2 \otimes \vect u_2)
      + \nabla \left(g\rho_1 h_1 h_2 + \frac 1 2 g \rho_2 h_2^2\right)
      &= \vect G_{h_2u_2} + g\rho_1h_1\nabla (h_2 + b) - g(\rho_1h_1+\rho_2h_2)\nabla b,
      \label{eq:2LSWEGeneral_mom2}
    \end{align}
  \end{subequations}
\end{widetext}
where the subscripts 1 and 2 denote the top and bottom layers respectively.
The coupling between the two layers are captured by the last source terms on the right-hand side of the momentum equations.

This form was originally described by \citeauthor{ovsyannikov1979},\cite{ovsyannikov1979} and is referred to in more recent works as ``the conventional two-layer shallow-water model''.\citep{chiapolino2018}

\subsection{2LSWE forms that are reducible to one-layer approximations}
Conservation of momentum can be considered at three different scales:
\begin{enumerate}
  \item The globally conservative form where total momentum is conserved.
  \item The layerwise conservative form (\cref{eq:2LSWEGeneral}) where the momentum in each layer is conserved.
  \item The locally conservative form where the local momentum, or velocity, is conserved.
\end{enumerate}
Although these formulations are equivalent for smooth solutions, they are not generally equivalent, as will be further discussed in \cref{sec:RankineHugoniot}.
The layerwise formulation is not easily reducible to a one-layer model.
The remaining two approaches can be converted to an effective one-layer approximation, and our analysis will cover both.
In the locally conservative form, we combine \cref{eq:2LSWEGeneral_mom1,eq:2LSWEGeneral_mom2} with \cref{eq:2LSWEGeneral_h1,eq:2LSWEGeneral_h2} to give equations for velocity rather than momentum.
Using the product rule for differentiation,
\[
    \div \left(\rho_ih_i\vect u_i \otimes \vect u_i\right)
    = \vect u_i \div (\rho_ih_i\vect u_i) + \rho_ih_i\left(\vect u_i \cdot \nabla\right)\vect u_i,
\]
we arrive at the set of equations which we shall refer to as the \emph{locally conservative} version of the 2LSWE,
\begin{widetext}
  \begin{subequations}
    \label{eq:2LSWEGenVel}
    \begin{align}
      \pd t \rho_1h_1 + \div (\rho_1h_1\vect u_1) &= G_{h_1}, \\
      \pd t \rho_2h_2 + \div (\rho_2h_2\vect u_2) &= G_{h_2}, \\
      \pd t \vect u_1 + (\vect u_1 \cdot \nabla)\vect u_1
      + \nabla \left[g (h_1 + h_2 + b) \right]
      &= \frac{1}{\rho_1h_1}(\vect G_{h_1u_1} - \vect u_1 G_{h_1}),
      \label{eq:2LSWEGenVel_u1} \\
      \pd t \vect u_2 + (\vect u_2 \cdot \nabla)\vect u_2
      + \nabla \left[g \left(\frac{\rho_1}{\rho_2}h_1 + h_2 + b\right)\right]
      &= \frac{1}{\rho_2h_2}(\vect G_{h_2u_2} - \vect u_2 G_{h_2}).
      \label{eq:2LSWEGenVel_u2}
    \end{align}
  \end{subequations}
\end{widetext}
For a comprehensive study of the well-posedness of the locally conservative 2LSWE, see for instance.\citep{monjarret2015}

When conserving the total momentum, the sum of \cref{eq:2LSWEGeneral_mom1} and \cref{eq:2LSWEGeneral_mom2} is used, which has the advantage of eliminating the interaction between the layers.
However, this approach requires an additional conservation law.
\citet{ostapenko1999,ostapenko2001} showed that the additional conservation law should be the difference between \cref{eq:2LSWEGenVel_u2} and \cref{eq:2LSWEGenVel_u1}.
\citeauthor{ostapenko2001} used these equations in a study of the well-posedness of the 2LSWE.
The resulting equations, which we will refer to as the \emph{globally conservative} version of the 2LSWE, read
\begin{widetext}
\begin{subequations}
  \label{eq:2LSWEGenTot}
  \begin{gather}
    \pd t \rho_1h_1 + \div (\rho_1h_1\vect u_1) = G_{h_1},
    \label{eq:2LSWEGenTot_h1} \\
    \pd t \rho_2h_2 + \div (\rho_2h_2\vect u_2) = G_{h_2},
    \label{eq:2LSWEGenTot_h2} \\
    \label{eq:2LSWEGenTot_sumMom}
    \begin{aligned}[b]
      \pd t \left(\rho_1 h_1 \vect u_1 + \rho_2 h_2 \vect u_2\right)
      &+ \div \left(\rho_1 h_1 \vect u_1 \otimes \vect u_1
      + \rho_2 h_2 \vect u_2 \otimes \vect u_2\right)\\
      &\qquad+ \nabla \left(
        \frac 1 2 g\rho_1 h_1^2 + \rho_1 g h_1 h_2
      + \frac 1 2 \rho_2 g h_2^2\right)
      = \vect G_{h_1u_1}+\vect G_{h_2u_2} - g(\rho_1h_1+\rho_2h_2)\nabla b,
    \end{aligned}\\
    \label{eq:2LSWEGenTot_diffVel}
    \pd t \left(\vect u_2 - \vect u_1\right)
    + (\vect u_2 \cdot \nabla)\vect u_2 - (\vect u_1 \cdot \nabla)\vect u_1
    - \nabla \left(g \delta h_1\right) = \vect J,
  \end{gather}
\end{subequations}
\end{widetext}
where
\[
  \vect J = \frac{\vect G_{h_2u_2} - \vect u_2 G_{h_2}}{\rho_2h_2}
  - \frac{\vect G_{h_1u_1} - \vect u_1 G_{h_1}}{\rho_1h_1}
\]
and where we have defined
\begin{equation}
  \delta \defeq \frac{\rho_2 - \rho_1}{\rho_2}.
  \label{eq:delta}
\end{equation}

\subsection{The Rankine-Hugoniot condition}
\label{sec:RankineHugoniot}
When two sets of equations are equivalent in the classical sense, they may not be equivalent in the weak sense, that is, when interpreted as distributions.\citep{whitham1974,holden2015,borthwick2016}
In the 2LSWE, \cref{eq:2LSWEGeneral,eq:2LSWEGenVel,eq:2LSWEGenTot} are equivalent for smooth solutions, but not for weak solutions.
In particular, these equations will give different shock velocities.
We shall next discuss the mathematical framework used to analyze such discontinuities; the Rankine-Hugoniot condition, named after \citeauthor{rankine1870} and \citeauthor{hugoniot1887} who first introduced it.\citep{rankine1870,hugoniot1887,hugoniot1889}

The Rankine-Hugoniot condition states the following.
Assume that $u$ satisfies a general scalar conservation equation
\begin{equation}
  \pd{t} u(t, \vect x) + \div \vect q (t, \vect x) = J
\end{equation}
in the weak sense, where $J$ is some source term that does not involve the derivatives of $u$.
Further, assume that $u$ has a discontinuity along some curve $\Gamma$.
For any function $f$, define the jump across a discontinuity as $\jump{f} \equiv f_r - f_l$, where $f_r \equiv \lim_{\varepsilon \to 0^+}f(\vect \xi+ \varepsilon \vect{\hat n})$ and $f_l \equiv \lim_{\varepsilon \to 0^-}f(\vect \xi+ \varepsilon \vect{\hat n})$.
The Rankine-Hugoniot condition then states that the discontinuity at any point $\vect \xi \in \Gamma$ propagates along the outward-pointing normal vector $\vect{\hat n}$ with a speed $S$.
This speed is called the shock speed and satisfies the relation
\begin{equation}
  S\jump{u} = \nd \jump{\vect q}.
  \label{eq:rankineHugoniotScalar}
\end{equation}
Similarly, if $\vect u$ satisfies a general vector conservation equation,
\begin{equation}
  \pd{t} \vect u(\vect x,t) + \div(\vect a \otimes \vect b) + \nabla q (\vect x,t) = \vect J,
\end{equation}
then, if there is some discontinuity in $\vect u$, we have the result
\begin{equation}
  S\jump{\vect u} = \jump{\n \cdot (\vect a \otimes \vect b) + q\vect{\hat n}}.
  \label{eq:rankineHugoniotVector}
\end{equation}

\Cref{eq:rankineHugoniotScalar,eq:rankineHugoniotVector} can be directly applied to the mass conservation equations and the conservation law for total momentum, respectively.
In one dimension, the Rankine-Hugoniot conditions can also be applied to the locally conservative momentum equation.
In two dimensions, the term $\vect u\cdot \nabla \vect u$ renders the Rankine-Hugoniot condition for the transversal velocity component ill-defined.
Nevertheless, for our purposes we do not need the Rankine-Hugoniot condition for the transversal velocity component.
See \cref{app:travellingWave} for a discussion on this.
In the layerwise momentum equation, the interaction term $\propto h_1\nabla h_2$ makes the normal component for the momentum equations ill-defined, which is why we must exclude this formulation of the 2LSWE from the analysis.

We derive the Rankine-Hugoniot conditions in \cref{app:travellingWave} and find that for the locally conservative 2LSWE~(\cref{eq:2LSWEGenVel_u1,eq:2LSWEGenVel_u2}),
\begin{equation}
  S\jump{\vect u_i}\cdot\n = \jump{\frac 1 2 (\ndot u_i)^2 + g \left(\frac{\rho_1}{\rho_2}\right)^{i-1}h_1 + gh_2},
  \label{eq:2LSWEGenVelRH}
\end{equation}
where as before $i=1,2$ denotes the layer.
Similarly, for the globally conservative 2LSWE~\cref{eq:2LSWEGenTot_sumMom,eq:2LSWEGenTot_diffVel}, we find
\begin{multline}
  S\jump{\rho_1h_1\vect u_1 + \rho_2h_2\vect u_2} \\
  = \jump{(\ndot u_1)\rho_1h_1\vect u_1
    + (\ndot u_2)\rho_2h_2\vect u_2} \\
    + \jump{\frac 1 2 g \rho_1 h_1^2 + \rho_1 g h_1 h_2
      + \frac 1 2 \rho_2 g h_2^2}\n
  \label{eq:2LSWEGenTotRH_totMom}
\end{multline}
and
\begin{multline}
  S\nd\jump{\vect u_2 - \vect u_1} =
  \Biggl\llbracket
    \frac 1 2\biggl[(\ndot u_2)^2 \\
    - (\ndot u_1)^2\biggr] - g\delta h_1 \Biggr\rrbracket.
  \label{eq:2LSWEGenTotRH_diffVel}
\end{multline}

Finally, we note that in calculations with the Rankine-Hugoniot condition it is useful to observe that $\jump{ab} = \jump{a}\av{b} + \av{a}\jump{b}$ where $\av{a} = (a_l + a_r)/2$.

\subsection{Physical solutions}
\label{sec:physicalSolutions}
When PDEs are considered in the weak sense, it is necessary to impose extra conditions to extract a unique physical solution.
Such conditions are called \emph{entropy conditions}.
In this subsection, we will introduce one such condition: the \emph{energy requirement}.
For simplicity, we define a \emph{physical solution} as one that satisfies the energy requirement.

The energy requirement states that only shocks that dissipate energy are physical.
This translates into requiring that the energy of the physical solution does not increase in time except from possible source terms.
Energy, in this sense, has the role of a mathematical entropy.\cite{holden2015}
However, the word \emph{entropy} is typically restricted to convex functions of the solution variables.
As has been showed by \citeauthor{ostapenko1999},\cite{ostapenko1999} energy is indeed a convex function of the globally conservative system, but for the locally conservative system it is convex only for subcritical flow.
Because we here cover both cases we use the word energy rather than entropy.

The energy of the 2LSWE reads
\begin{multline}
  E = \frac 1 2 \left(\rho_1 h_1 \abs{\vect u_1}^2 + \rho_2 h_2 \abs{\vect u_2}^2\right)\\
  + g\Biggl[\rho_2h_2\left(\frac 1 2 h_2 + b\right)
  +  \left(\frac 1 2 h_1+h_2+b\right) \rho_1h_1 \Biggr].
  \label{eq:energy}
\end{multline}
This expression is given in terms of parameters that are already solved for in the 2LSWE.
For smooth solutions we may therefore combine the subequations of the 2LSWE to form a conservation law for the energy.
By exchanging the equality in this conservation law by an inequality, it can be fulfilled also by weak, discontinuous solutions.
We obtain
\begin{widetext}
  \begin{multline}
    \pdfun E t + \div \left[\vect q_1\left(g\left(h_1 + h_2 + b\right)
      + \frac 1 2 \abs{\vect u_1}^2\right)
      + \vect q_2\left(g\left(\frac{\rho_1}{\rho_2}h_1 + h_2 + b\right)
    + \frac 1 2 \abs{\vect u_2}^2\right)\right] \\
    \le \vect u_1 \cdot \vect G_{h_1u_1}
    + \vect u_2 \cdot \vect G_{h_2u_2}
    - \frac 1 2 g h_1^2 \pdfun{\rho_1}{t}
    - g h_2 \left(\frac{\rho_1}{\rho_2}h_1
      + \frac 1 2 h_2 \right) \pdfun{\rho_2}{t} \\
    + G_{h_1}\left(g\left(h_1 + h_2 + b\right) - \frac 1 2 \abs{\vect u_1}^2\right)
    + G_{h_2}\left(g\left(\frac{\rho_1}{\rho_2}h_1 + h_2 + b\right)
      - \frac 1 2 \abs{\vect u_2}^2\right),
    \label{eq:energyConservation}
  \end{multline}
\end{widetext}
where $\vect q_i = \rho_ih_i\vect u_i$ for short.

\section{Derivation of Froude numbers from the 2LSWE}
\label{sec:BenjaminAlternative}
In the following, we briefly illustrate the surprising effectiveness of the 2LSWE to predict shock speeds despite its underlying assumption of negligible vertical acceleration.
To do this we apply the Rankine-Hugoniot conditions and the 2LSWE to derive expressions for the leading edge Froude number ($\FrLe$) of two-layer systems with fixed total height.

Shock speeds in two-layer systems with fixed total height is important for instance in lock-exchange and lock-release problems, where a heavy fluid is spreading within a lighter fluid inside a rectangular channel as illustrated in \cref{fig:spreading-channel}.
Such problems have been studied extensively, and there is a large number of results from laboratory experiments available.\citep{rottman1983,huppert1980,shin2004}
Moreover, much theoretical work has been carried out to model the Froude-number for flows inside rectangular channels,\citep{benjamin1968,priede2019,borden2013} which means that this is a good candidate for testing the credibility of shock behaviour in the 2LSWE.

Most previous works have focused on fluids with similar densities such that $\delta \ll 1$ for $\delta$ given by \cref{eq:delta}.
This is referred to as the Boussinesq case.\citep{boussinesq1903}
The most commonly used front condition applied to such flows is the equation for the Froude number given by \citeauthor{benjamin1968},\cite{benjamin1968} \cref{eq:benjamin}.
For instance, \citet{ungarish2011} has applied the Froude number by Benjamin as a boundary condition when solving the 2LSWE for rectangular geometries.
They also generalized this to arbitrary geometries.\citep{ungarish2013}
\begin{figure}
  \centering
  \includegraphics[width=0.7\columnwidth]{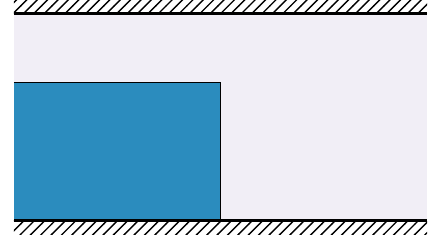}
  \caption{A sketch of the initial condition for the lock-exchange problem: Two-layer shallow-water flow in a rectangular channel.
  The grey fluid is lighter than the blue, and the initial shock is the vertical line beween blue and grey.}
  \label{fig:spreading-channel}
\end{figure}

For the particular problem where the two-layer flow is confined inside a rectangular channel, the sum of the layer depths must be constant; $h_1 + h_2 = H$.
In this case, there are no free surfaces.
We therefore add an additional pressure term $p_0$ that may vary in time and space but is constant in the vertical direction.

We first consider the locally conservative 2LSWE~\eqref{eq:2LSWEGenVel} in one spatial dimension with the added free pressure term,
\begin{subequations}
  \label{eq:2LSWECavity}
  \begin{gather}
    \pdt{h_1} + \pd{x}(h_1u_1) = 0,
    \label{eq:2LSWECavityh1}\\
    \pdt{h_2} + \pd{x}(h_2u_2) = 0,
    \label{eq:2LSWECavityh2}\\
    \pdt{u_1} + \pd{x}\left(\frac{1}{2}u_1^2 + gh_2 + \frac{1}{\rho_1} p_0\right) = 0,
    \label{eq:2LSWECavityu1}\\
    \pdt{u_2} + \pd{x}\left(\frac{1}{2}u_2^2 + gh_2 + \frac{1}{\rho_2} p_0\right) = 0.
    \label{eq:2LSWECavityu2}
  \end{gather}
\end{subequations}
These are the same equations that were used by \citet{rottman1983} to study spreading of gravity currents.
\citeauthor{rottman1983} added \cref{eq:benjamin} as an additional equation for the Froude number at the leading edge, but in the following we will show that a similar expression for $\FrLe$ can be obtained from \cref{eq:2LSWECavity} directly.

With $h_1 + h_2$ constant, the sum of \cref{eq:2LSWECavityh1,eq:2LSWECavityh2} implies that $h_1u_1 + h_2u_2$ is constant in $x$.
If we assume that the total momentum is $0$ at the boundary, \eg due to a wall or because the boundary is at infinity and the fluids were initially at rest, we may set $h_1u_1 + h_2u_2 = 0$.

By use of the Rankine-Hugoniot condition~\eqref{eq:rankineHugoniotScalar} to \cref{eq:2LSWECavityh1,eq:2LSWECavityh2}, we get
\[
  S = u_{2,l} = -\frac{h_{1,l}}{h_{2,l}} u_{1,l},
\]
where, as before, the subscript $l$ indicates the left side of the shock.
We next apply the Rankine-Hugoniot condition to~$\rho_2$\eqref{eq:2LSWECavityu2}~$-$~$\rho_1$\eqref{eq:2LSWECavityu1}, which gives
\begin{multline}
  S\left(\rho_2 S + \rho_1\frac{h_{2,l}}{h_{1,l}}S\right)
  = \frac{1}{2} \left(\rho_2 S^2 - \rho_1 \frac{h_{2,l}^2}{h_{1,l}^2}S^2\right) \\
  + \rho_2 gh_{2,l} - \rho_1 gh_{2,l}.
\end{multline}
After some algebraic manipulation, we find that
\begin{equation}
  \FrLe^2 = \frac{u_{2,l}^2}{g\delta h_{2,l}}
  = \frac{2 (1 - \alpha)^2}{1 - \delta \alpha(2 - \alpha)}
  \label{eq:froudeCavity}
\end{equation}
where $\alpha = h_{2,l}/(h_{2,l} + h_{1,l})$.

We next consider the globally conservative 2LSWE~\eqref{eq:2LSWEGenTot}.
A similar analysis and derivation now gives
\begin{equation}
  \FrLe^2 = \frac{2(1-\alpha)^2(1-\delta\alpha/2)}{1 - 2\delta\alpha(1-\alpha)}.
  \label{eq:froudeCavityTotMom}
\end{equation}
As expected, the different formulation of the 2LSWE leads to a different expressions for the Froude number.

The Boussinesq approximation is achieved by setting $\delta=0$ wherever it is not multiplied by $g$.
In this case \cref{eq:froudeCavity,eq:froudeCavityTotMom} coincide and gives that
\begin{equation}
  \FrLe = \sqrt{2} (1 - \alpha).
  \label{eq:Fr2LSWEBous}
\end{equation}
\Cref{fig:BenjaminVsUs} compares our results from the 2LSWE, \cref{eq:Fr2LSWEBous}, to the model by \citeauthor{benjamin1968},\cite{benjamin1968} \cref{eq:benjamin}.
As can be seen, the difference is small.
\Cref{eq:benjamin} is obtained by balancing forces and does not rely on any assumptions regarding negligible vertical velocities.
The similarity of \cref{eq:benjamin,eq:Fr2LSWEBous} therefore indicates that the breakdown of the shallow-water equations in vicinity of shocks is not so severe as one would think and as has been repeatedly assumed in the literature.\citep{hatcher2014,rottman1983,hoult1972,fannelop1972}
\begin{figure}
  \centering
  \tikzsetnextfilename{BenjaminVsUs}
  \begin{tikzpicture}[
      benjamin/.style={
        color=Dark2-A,
      },
      priede/.style={
        dashdotted,
        color=Dark2-H,
      },
      delta00/.style={
        densely dashed,
        color=Dark2-D,
      },
      delta04/.style={
        dashed,
        color=Dark2-B,
      },
      delta09/.style={
        densely dotted,
        color=Dark2-E,
      },
    ]
    \begin{axis}[
        width = 0.9*\columnwidth,
        height = 0.5*\columnwidth,
        ymin = 0.6,
        ymax = 1.5,
        xmin = 0.0,
        xmax = 0.525,
        xtick={0, 0.1, 0.2, 0.3, 0.4, 0.5},
        xticklabels={0, 0.1, 0.2, 0.3, 0.4, 0.5},
        ytick={1/sqrt(2), 1.0, sqrt(2)},
        yticklabels={$1/\sqrt{2}$, 1, $\sqrt{2}$},
        axis lines = left,
        grid = both,
        no marks,
        domain = 0:0.5,
        samples = 250,
        every axis plot/.append style={line width=1pt},
        xlabel = $\alpha$,
        ylabel = $\FrLe$,
        ylabel style={align=center},
      ]

      \newcommand*{\FrLeLoc}[1]{sqrt(2*(1-x)^2/(1-2*#1*x+#1*x^2))}
      \newcommand*{\FrLeGlob}[1]{sqrt(2*(1-x)^2*(1-#1*x/2)/(1-2*#1*x*(1-x)))}
      \newcommand*{\FrLeBenj}{sqrt((2-x)*(1-x)/(1+x))}
      \newcommand*{\FrLePriede}{sqrt((1 - x)^3/(1/2 - x^2))}

      \addplot+[benjamin]{\FrLeBenj};
      \label{plots:benjamin}
      \addplot+[delta00]{sqrt(2)*(1-x)};
      \label{plots:delta0}

      \legend{Benjamin, 2LSWE}
    \end{axis}

  \end{tikzpicture}
  \caption{Froude numbers calculated from the 2LSWE in the Boussinesq case (\cref{eq:Fr2LSWEBous}) compared to the equation by \citet{benjamin1968} (\cref{eq:benjamin}).}
  \label{fig:BenjaminVsUs}
\end{figure}
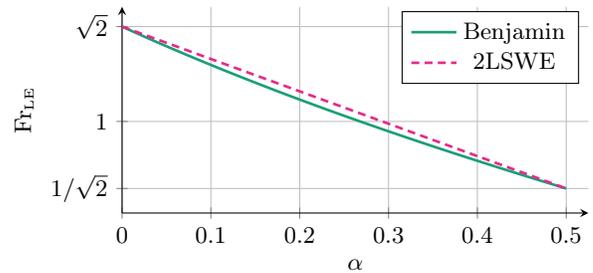

One advantage of the 2LSWE is that it does not use the Boussinesq approximation.
The non-Boussinesq case has more recently received attention in the literature,\cite{lowe2005,birman2005} and \cref{eq:froudeCavity,eq:froudeCavityTotMom} could be of interest in this regard.

The treatment presented here is under the assumption of negligible mixing between the layers.
In systems with mixing, \citet{sherwoods2015} has found the spreading is slower because the density difference at the leading edge, and hence the effective gravity, is reduced with time.
With their time-dependent reduced gravity, they found experimentally that $\FrLe = 0.90\pm 0.05$ for $\alpha = 0.37$.
Inserting $\alpha = 0.37$ into \cref{eq:Fr2LSWEBous} we get $\FrLe = 0.89$.
That is, if mixing is taken into account in the shallow water framework by introducing a slowly varying time-dependent density difference and possibly some source terms that do not affect the Rankine-Hugoniot condition, the resulting Froude number at the leading edge is in good agreement with the observed value.

Finally, we note that \citet{priede2019} has also found an expression for the Froude number in the 2LSWE with constant height.
They restricted the analysis to the Boussinesq case and got a result which differs slightly from \cref{eq:Fr2LSWEBous}.
The reason for the deviation is that they rewrote the equations in terms of new variables, $\eta \defeq h_1 - h_2$ and $\vartheta \defeq u_1 - u_2$, and used $\eta$ and $\eta\vartheta$ as conserved quantities before they applied the Rankine-Hugoniot condition.
This changes the weak solutions and hence the shock speed.

\section{Reducing the two-layer systems to effective one-layer systems}
\label{sec:reducingToOneLayer}
In this section, we present a theorem with a constructive proof that demonstrates that it is possible to reduce the 2LSWE into an effective one-layer model while preserving the correct behaviour of shocks and contact discontinuities.
The theorem shows that this decoupling is possible when the depth of one layer becomes large compared to the other layer.
We show that additional closures for the shock velocity are not needed, which differs from previous reductions to one-layer models presented in the literature.

\subsection{The constant-height lemma}
In the following, we denote by $s$ and $d$ the \emph{relatively} shallow and deep layers, respectively.
This means that with $(s,d) = (1,2)$, the top layer is shallow relative to the bottom layer, and vice versa for $(s,d) = (2,1)$.
Further, we let $\spaceav f$ denote the average of $f$ over the region in which it is defined.

In order to state and prove the theorem, we will use a concept we call \emph{source-boundedness}.
We will also use a lemma that states that in the indicated limits of the theorem, the relative height of the deepest layer does not change with time.

\begin{definition}[Source-boundedness]
  \normalfont
  Layer $i \in \{1,2\}$ in a two-layer shallow-water system is source-bounded if there exists $K \in \mathds{R}$ such that the source terms satisfy $\forall h_i > K$,
  \[
    \pd{h_i} \abs{\frac{G_{h_j}}{\rho_i h_i}} < 0
    \quad \text{ and } \quad
    \pd{h_i} \abs{\frac{\vect G_{h_ju_j}}{\rho_i h_i}} < 0,
  \]
  for $j = 1$ and $j=2$.
  \label{def:sourceBounded}
\end{definition}
\begin{lemma}
  \label{lemma:constantHeight}
  Let $(s, d) = (1, 2)$ or $(2, 1)$, $\{D_k\}_{k\in\mathds{N}}$ be a sequence of increasing real numbers, $h_0$ and $f$ be scalar functions, and $\vect q_{1,0}$ and $\vect q_{2,0}$ be vector functions.
  Further, consider a 2LSWE system with initial conditions
  \[
    \begin{aligned}
      h_{dk}(0,\vect x) &= D_k + f(\vect x), \\
      h_{sk}(0,\vect x) &= h_0(\vect x), \\
      \vect q_{1k}(0, \vect x) &=\vect q_{1,0}(\vect x), \\
      \vect q_{2k}(0, \vect x) &= \vect q_{2,0}(\vect x),
    \end{aligned}
  \]
  where layer $d$ is source-bounded and where both layer $d$ and the bottom layer (these are the same if $d=2$) have constant average density.
  Now let $\{(h_{1k}, h_{2k}, \vect q_{1k}, \vect q_{2k})\}_{k\in\mathds{N}}$ be physical solutions to the 2LSWE system.
  If $\{(h_{sk}, h_{dk} - D_k, \vect q_{sk}, \vect q_{dk}/D_k)\}_{k\in\mathds{N}}$ converge and the first and second derivatives are uniformly bounded in the regions where they are well-defined, then
  \[
    \lim_{k \to \infty}\frac{h_{dk}(t, \vect x)
    }{D_k} = 1.
  \]
\end{lemma}
\begin{proof}
  First, note that since the second derivatives are uniformly bounded, the mean value theorem implies that the first derivatives are equicontinuous.
  Then, since the first derivatives are also bounded, the Arzelà-Ascoli theorem gives that there is a subsequence where the first derivatives are uniformly convergent.\citep{dunford1957}
  This implies that we can interchange the order of limits and differentiation.\citep{rudin1976}
  From the definition of the energy in \cref{eq:energy} it is clear that all terms are non-negative.
  This, in addition to the fact that the energy is a convex function of the heights, means that for a system with constant bottom topography and $\spaceav{h_i} = 1$ for $i \in \{1,2\}$, the energy is bounded from below by the height and momentum distributions that give $E = g\rho_i/2$.
  We let $\escaled = 2E_k/D_k^2$ be a scaled energy, and it follows by insertion that $\escaled(0, \vect x) \to g\rho_d$ as $k \to \infty$.
  That is, the scaled energy $\escaled(0, \vect x)$ approaches the minimal for a system with $\spaceav{h_{dk}(t, \vect x)/D_k} = 1$ in the limit when $k\to\infty$.
  Source-boundedness of the mass source terms implies that $\spaceav{h_{dk}/D_k} \to 1$ for all $t \in \mathds{R}$ as $k \to \infty$.

  Further, since a physical solution must satisfy the energy conservation \eqref{eq:energyConservation}, it similarly follows by use of the source-boundedness that
  \[
    \pdfun {\escaled(0, \vect x)} t \leq 0
  \]
  in the limit when $k\to\infty$.
  Because all the terms in \cref{eq:energy} are non-negative and because the right hand side of the scaled version \cref{eq:energyConservation} remain 0 as long as the scaled energy remains minimal, we must have
  \begin{equation}
    \lim_{k \to \infty} \abs{\escaled(t, \vect x) - \escaled(0, \vect x)} = 0,
    \label{eq:energyRemainsMinimal}
  \end{equation}
  Assume that $\exists \varepsilon > 0$ and $\forall N \in \mathds{N}$, $\exists k > N$, such that
  \[
    \abs{\frac{h_{dk}(t, \vect x)}{D_k} - 1} > \varepsilon.
  \]
  This implies that $h_{dk}(t, \vect x)/D_k$ deviates from 1 by a term which does not vanish in the limit $k \to \infty$.
  This contradicts \cref{eq:energyRemainsMinimal}, as discussed above, so
  $h_{dk}(t, \vect x)/D_k \to 1$ for all $t \in \mathds{R}$ and $\vect x \in \mathds{R}^2$ as $k \to \infty$.
\end{proof}

\subsection{The one-layer approximation theorem}
In the following theorem, we show that in the similar limits as above, the 2LSWE may be reduced to the locally conservative 1LSWE~\eqref{eq:1LSWELoc} with a reduced gravity, $g\to\delta g$ with $\delta$ as defined in \cref{eq:delta}.
In the case where the top layer is shallow relative to the bottom layer, the bottom topography term drops out of the equation governing the top layer.
As before, we use $s\in\{1, 2\}$ to indicate which layer is shallow relative to the other, such that
\begin{subequations}
  \label{eq:1LSWELocGen}
  \begin{gather}
    \pd t \rho_sh_s + \div (\rho_sh_s\vect u_s) = G_{h_s}, \\
    \begin{multlined}[b]
      \pd t \vect u_s + (\vect u_s \cdot \nabla)\vect u_s
      + \delta g \nabla(h_s + b^{s-1})\\
      = \frac{1}{\rho_sh_s}(\vect G_{h_su_s} - \vect u_s G_{h_s}).
    \end{multlined}
  \end{gather}
\end{subequations}

\begin{theorem}
  \label{thm:oneLayerApproximation}
  Let $(s, d) = (1, 2)$ or $(2, 1)$, $\{D_k\}_{k\in\mathds{N}}$ be a sequence of increasing real numbers, $h_0$ and $f$ be scalar functions, and $\vect q_{1,0}$ and $\vect q_{2,0}$ be vector functions, all defined on $\Omega \subseteq \mathds{R}^n$.
  Further, consider a 2LSWE in the form of \cref{eq:2LSWEGenVel} or \cref{eq:2LSWEGenTot} with initial conditions
  \[
    \begin{aligned}
      h_{dk}(0,\vect x) &= D_k + f(\vect x), \\
      h_{sk}(0,\vect x) &= h_0(\vect x), \\
      \vect q_{1k}(0, \vect x) &=\vect q_{1,0}(\vect x), \\
      \vect q_{2k}(0, \vect x) &= \vect q_{2,0}(\vect x),
    \end{aligned}
  \]
  in which layer $d$ is source-bounded and the density of layer $d$ is constant.
  Now let $\{(h_{1k}, h_{2k}, \vect q_{1k}, \vect q_{2k})\}_{k\in\mathds{N}}$ be physical solutions to the 2LSWE such that $\vect q_{dk}$ satisfies the boundary condition
  \begin{equation}
    \abs{\vect q_{dk}(t, \vect x)} \le K
    \qquad\text{for}\qquad \vect x\in\partial\Omega,
    \label{eq:boundaryMomentum}
  \end{equation}
  with $K \in \mathds{R}$ independent of $k$ and where $\partial\Omega$ may be at infinity.

  If $\{(h_{sk}, h_{dk} - D_k, \vect q_{sk}, \vect q_{dk}/D_k)\}_{k\in\mathds{N}}$ converge and the first and second derivatives are uniformly bounded in the regions where they are well-defined, then $(h_s, \vect u_{sk}) \to (h, \vect u)$ where $(h, \vect u)$ solves \cref{eq:1LSWELocGen} in the weak sense, $\vect u_{dk}\to \vect 0$, and $h_{dk}-D_k \to \left(C - \rho_s^{d-1}h_s - \rho_2^{d-1}b\right)/\rho_d^{d-1}$, where $C$ is constant in space.
  If the domain on which the solution is defined is infinite in range or the mass source terms are zero, then $C$ is equal to $C = \spaceav{\left[\rho_s^{d-1}h_s + \rho_d^{d-1}f + \rho_2^{d-1}b\right]}_{t=0}$.
\end{theorem}
\begin{proof}
  First, we note that weak solutions of \cref{eq:2LSWEGenVel,eq:2LSWEGenTot} will be piecewise differentiable and their states on both sides of a discontinuity are connected by a \emph{Hugoniot locus}.
  A Hugoniot locus at some location in phase space is defined as all those states for which there is a shock speed that satisfies the Rankine-Hugoniot condition.\citep{holden2015}

  To prove the theorem, it is therefore sufficient to show i) local convergence for regions where the solution is differentiable and ii) that the states that are allowed by the Hugoniot loci of the 2LSWE~(\eqref{eq:2LSWEGenVel} and~\eqref{eq:2LSWEGenTot}) converge to those of the 1LSWE~\eqref{eq:1LSWELocGen}.
  As before, we may interchange the order of limits and differentiation since the second derivatives are uniformly bounded.

  We will first prove i).
  This will be done by proving that $\vect u_{dk} \to \vect 0$ in the limit $k\to \infty$ by the use of the fundamental theorem of geometric calculus.
  The reader is referred to \citet{doran2003} for an overview of this branch of mathematics.
  The purpose of using this theorem is to give a way to explicitly express a vector quantity in terms of its divergence, curl and boundary conditions.

  From conservation of mass and through source-boundedness and \cref{lemma:constantHeight}, we get that
  \begin{align}
    \div \left(\frac{\vect q_{dk}}{D_k}\right)
    = \frac{G_{h_d}}{D_k} -\pd t \frac{\rho_{dk} h_{dk}}{D_k}
    \xrightarrow{k\to \infty} 0 \notag \\
    \implies \div \vect u_{dk} \xrightarrow{k\to\infty} 0.
    \label{eq:udiv}
  \end{align}
  Next, we show that also the curl of $\vect u_{dk}$ vanish in the limit $k\to\infty$.
  In the following, we use $ A\wedge B$ to denote the wedge product, or exterior product, of $A$ and $B$, and $A\odot B$ to denote their geometric product.
  Applying $\nabla\wedge$, a generalized curl, from the left of the velocity equation of \cref{eq:2LSWEGenVel} yields
  \begin{multline}
    \pdfun{w_{dk}}{t}
      + \nabla\wedge I^{-1}\odot\vect u_{dk} \wedge I^{-1}\odot w_{dk} \\
    = \frac{\nabla\wedge\left(\vect G_{h_du_d}
      - \vect u_{dk} G_{h_d}\right)}{\rho_{dk}h_{dk}},
    \label{eq:ucurl}
  \end{multline}
  where $w_{dk} \equiv \nabla\wedge\vect u_{dk}$ is a bivector which is equal in magnitude to the curl of $\vect u_{dk}$ but well-defined in any dimension.
  Here $I^{-1} = \vect e_m \wedge \dotsb \wedge \vect e_1$, where $\vect e_i$ is the unit vector in direction $i$ and $m$ is the number of dimensions.
  Source-boundedness and the fact that $w_{dk} = 0$ at $t=0$ implies that $w_{dk}\to \vect 0$ in the limit $k\to\infty$.

  From the Helmholtz theorem, we know that a vector field defined on a finite domain or which goes sufficiently fast to $\vect 0$ is uniquely specified by its boundary condition, curl and divergence.
  Using techniques from geometric calculus, we can give an analytic expression.
  The fundamental theorem of geometric calculus states that\citep{doran2003}
  \begin{align}
    \oint_{\partial V} L(I_{m-1}(\vect x'))\diff^{m-1}x'  \int_V \dot{L}(\dot{\nabla}\odot I_m) \diff^{m}x',
    \label{eq:fundamentalThm}
  \end{align}
  where $L$ is any linear function,
  $I_m$ is the pseudoscalar of the tangent space to $V$ and $I_{m-1}(\vect x)$ is the pseudoscalar of the tangent space to $\partial V$ at $\vect x$.
  The vector derivative is $\nabla\odot A = \nabla\cdot A + \nabla \wedge A$, and the overdot indicates where it acts.
  That is, the integrand on the right hand side of \cref{eq:fundamentalThm} is $\sum_i\partial_i L(\vect e_i\odot I_m)$.
  See for instance the textbook by \citeauthor{doran2003}.\cite{doran2003}

  Let $V$ be some region where $u_{dk}$ is differentiable and let
  \begin{equation}
    L(A) = G\odot A \odot\vect u_{dk} =
    \frac{\vect x' - \vect x}{S_{m-1}\abs{\vect x' - \vect x}^m}\odot A \odot \vect u_{dk},
    \label{eq:L}
  \end{equation}
  where $G$ is the Green's function for the vector derivative, meaning that $\nabla\odot G(x', x) = \delta(x- x')$, and $S_{m-1}$ is the volume of the $(m-1)$-sphere.
  \Cref{eq:fundamentalThm} then states that
  \begin{widetext}
  \begin{multline}
    \vect u_{dk}(\vect x) = \frac{I_m^{-1}}{S_{m-1}}
    \odot\Biggl\{
      (-1)^m\int_V\frac{\vect x' - \vect x}{\abs{\vect x' - \vect x}^m}
      \odot I_m\odot \left[
        \nabla\cdot\vect u_{dk}(\vect x')
        + w_{dk}(\vect x')
      \right]\diff^mx' \\
      +\oint_{\partial V}\frac{\vect x' - \vect x}{\abs{\vect x' - \vect x}^m}
      \odot I_{m-1}(\vect x')\odot \vect u_{dk}(\vect x')
      \diff^{m-1}x'
    \Biggr\},
    \label{eq:udAnalytic}
  \end{multline}
  \end{widetext}

  The surface integral can be decomposed into one vector component whose integrand is proportional to $\vect u_{dk}\cdot\n$ and one triplet-vector component whose integrand is proportional to $\vect u \wedge \n$.
  Only the vector component will contribute to $\vect u_{dk}$.
  To show that $\vect u_{dk} \to 0$, it remains only to show that the last integral in \cref{eq:udAnalytic} goes to zero.
  This is proved in part ii) by showing that $\lim_{k\to\infty}\ndot u_{dk}$ is continuous.
  By applying \cref{eq:fundamentalThm} with $L$ given by \cref{eq:L} on a domain which does not include $\vect x$ the only contribution comes from surface integral in the limit $k\to \infty$, because $\lim_{k\to\infty} \nabla\cdot\vect u_{dk} + w_{dk} = 0$ everywhere.
  $I_{m-1}$ has opposite sign on opposite sides on surfaces, so surface integrals from neighbouring domains cancel as $\lim_{k\to\infty} \vect u_{dk}$ is continuous.
  Thus, we can extend the integral over $\partial V$ to an integral over $\partial \Omega$ by applying the fundamental theorem of geometric calculus in the neighboring domains.
  From \cref{eq:boundaryMomentum} with lemma 1, we get that $\lim_{k\to\infty}\vect u_{dk}$ must vanish on $\partial\Omega$.
  Hence, $\lim_{k\to\infty}\vect u_{dk} = \vect 0$ everywhere.

  From the momentum equations in \cref{eq:2LSWEGenVel}, then,
  \begin{multline}
    \nabla \left[\rho_{1k}^{d-1}h_{1k} + \rho_{2k}^{d-1}(h_{2k} + b)\right]
    = \Biggl(\frac{\vect G_{h_du_d} - \vect u_{dk} G_{h_d}}{\rho_{dk}h_{dk}} \\
    - \pdfun{\vect u_{dk}}{t}
    - (\vect u_{dk} \cdot \nabla)\vect u_{dk}\Biggr)\rho_{dk}^{d-1}
    \to \vect 0,
  \end{multline}
  and so in the limit $k\to\infty$, $\rho_s^{d-1}h_s + \rho_d^{d-1}(h_{dk} - D_k) + \rho_2^{d-1}b$ is constant in space.
  Finally, plugging this into the equation for $\vect u_s$ in \cref{eq:2LSWEGenVel} or \cref{eq:2LSWEGenTot}, we get \cref{eq:1LSWELocGen}.
  In the regions where the solution is differentiable, the various formulations of the 2LSWE, \cref{eq:2LSWEGeneral,eq:2LSWEGenVel,eq:2LSWEGenTot}, are equivalent.
  This completes the proof of i).

  For the proof of ii), we will compare the Hugoniot loci of the 2LSWE in the limit $k\to\infty$ to the Hugoniot locus of \cref{eq:1LSWELocGen}.
  Let $\gamma \defeq 1/\av{h_d}$.
  In \cref{app:2LSWE-Rankine-Hugoniot}, we show that the full set of Rankine-Hugoniot conditions for the 2LSWE \cref{eq:2LSWEGenVel,eq:2LSWEGenTot} may be written as
  \begin{subequations}
    \begin{align}
      S\jump{\rho_s h_s} &= \nd\jump{\rho_s h_s\vect u_s}, \\
      \nonumber
      S\nd\jump{\vect u_s} &= \jump{\frac 1 2 (\ndot u_s)^2 + \delta g h_s} \\
                           &+ g_1(\gamma, S, h_s, \ndot u_s, \ndot u_d), \\
      \jump{\rho_1^{d-1}h_1 + \rho_2^{d-1}h_2} &= g_2(\gamma, S, h_s, \ndot u_s, \ndot u_d), \\
      S\nd \jump{\vect u_d} &= g_3(\gamma, S, h_s, \ndot u_s, \ndot u_d),
    \end{align}
    \label{eq:2LSWECombinedRH}
  \end{subequations}
  where
  \begin{equation}
    g_3 = \gamma S \jump{h_d} \left(S - \av{\ndot u_d}\right).
  \end{equation}
  For \cref{eq:2LSWEGenVel},
  \begin{align}
    g_1 &= \gamma \jump{h_d}(S - \av{\ndot u_d})^2, \\
    g_2 &= \frac{\rho_2^{d-1}}{g} g_1.
  \end{align}
  For \cref{eq:2LSWEGenTot} with $d=1$,
  \begin{equation}
    g_1 = \left(\gamma\jump{h_1}(S - \av{\ndot u_1})^2
      + \delta g\, g_2\right)\n
  \end{equation}
  and
  \begin{multline}
    g_2 = \frac{\gamma}{g}\Biggl(
      \frac{S}{\rho_1}\jump{\rho_2h_2\ndot u_2}
      - \frac{1}{\rho_1}\jump{\rho_2h_2(\ndot u_2)^2} \\
      + \jump{h_1}\left(S\av{\ndot u_1} - \av{(\ndot u_1)^2}\right) \\
      - g\av{h_2}\jump{h_1} - \frac{\rho_2}{2\rho_1}\jump{h_2^2} \\
    + \jump{h_1}\left(S-\av{\ndot u_1}\right)\left(S - 2\av{\ndot u_1}\right)\Biggr).
  \end{multline}
  And finally, for \cref{eq:2LSWEGenTot} with $d=2$,
  \begin{equation}
    g_1 = \gamma \jump{h_2}(S - \av{\ndot u_2})^2
  \end{equation}
  and
  \begin{multline}
    g_2 = \frac{\gamma}{g}\Biggl(
      S\jump{\rho_1h_1\ndot u_1}
      - \jump{\rho_1h_1(\ndot u_1)^2} \\
      + S\rho_2 \jump{h_2}\av{\ndot u_2}
      - \rho_2\jump{h_2}\av{(\ndot u_2)^2} \\
      - \jump{\frac 1 2 g \rho_1 h_1^2}
      - \rho_1g\av{h_1}\jump{h_2} \\
    + \rho_2\jump{h_2}(S - \av{\ndot u_2})(S - 2 \av{\ndot u_2})\Biggr).
  \end{multline}
  In particular, we note that $g_1, g_2$, and $g_3$ vanish for $\gamma = 0$ in all cases.

  Next, we notice that \cref{eq:2LSWECombinedRH} with $\gamma = 0$ is exactly the Rankine-Hugoniot relations for the locally conservative 1LSWE~\eqref{eq:1LSWELocGen} together with the conditions that $\rho_1^{d-1}h_1 + \rho_2^{d-1}h_2$ is constant and $\vect u_d = \vect 0$.
  From \cref{lemma:constantHeight}, it follows that $\lim_{k\to\infty} \gamma = 0$.
  Thus the Hugoniot loci match, and this concludes the proof of the theorem.
\end{proof}

\subsection{Discussion of the theorem}
\Cref{thm:oneLayerApproximation} shows that we may approximate the thinnest layer of the 2LSWE with the locally conservative 1LSWE where $g \to (1 - \rho_1/\rho_2)g$ according to \cref{eq:1LSWELocGen}.
The approximation becomes more accurate when the depth of the deepest layer is increased without increasing momentum or other key properties.
\Cref{fig:theorem-cases} shows a sketch of how the two-layer cases converge to one-layer cases when we increase the ``depth'', $D_k$.
\begin{figure}
  \centering
  \includegraphics[width=1.0\columnwidth]{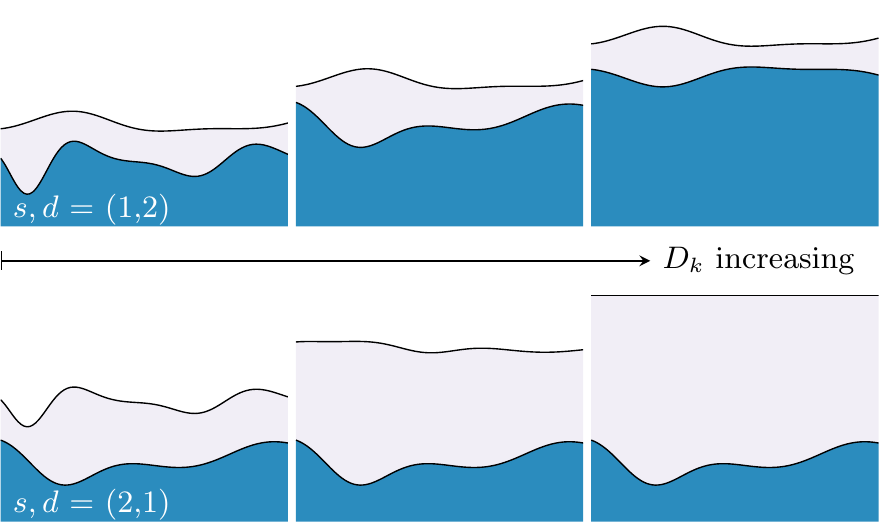}
  \caption{A sketch of how the two layers converge to one-layer cases with increasing $D_k$ for both of the cases $(s,d)=(1,2)$ and $(s,d)=(2,1)$.}
  \label{fig:theorem-cases}
\end{figure}

The interesting part about \cref{thm:oneLayerApproximation} is not that smooth solutions of the 2LSWEs can be approximated by solutions of the 1LSWE.
It is rather that a particular form of the 1LSWE, the locally conservative form, also captures weak solutions, meaning that it gives the correct shock speeds and relations between height- and velocity-distributions at either sides of discontinuities.
This is important, because while 1LSWE has been used to model two-layer spreading before, it has always been under the assumption that one must use additional equations at discontinuities in order to account for the effects of the additional layer.

A surprising implication of this result is that it suggests that the shallow water framework works better to describe shocks than one would anticipate from the assumption of negligible vertical acceleration.
By analytical and experimental considerations not related to the shallow water framework, it has been found that the Froude number at the leading edge of a spreading fluid in a two-layer system lies in the range $[1, \sqrt{2}]$.\citep{benjamin1968,ungarish2017,fay2007,fannelop1972,hoult1972,hatcher2014}
\Cref{thm:oneLayerApproximation} implies that this is also true in the shallow water model.
Using the Rankine-Hugoniot condition of the locally conservative 1LSWE we get $\FrLe = \sqrt{2}$.

From a practical standpoint, the result presented in this paper makes it more straightforward to use the shallow-water framework to model two-layer flow with discontinuous distributions, such as oil-spills.
Previous numerical schemes which have been created to ensure that the height- and velocity-distributions satisfy front conditions, which typically involves $\FrLe$, have had to track the position of the leading edge and alter the solution.\citep{hatcher2014,hatcher20013}
In contrast, when using the 1LSWE, which correctly captures shocks of 2LSWE, one automatically obtains numerical solutions that satisfy the Rankine-Hugoniot conditions and hence satisfies the front-condition $\FrLe = \sqrt{2}$.

Finally, we remark that the mathematical tools used to prove \cref{thm:oneLayerApproximation} are not directly applicable to the layerwise formulation of the 2LSWE~\eqref{eq:2LSWEGeneral}.
One way to possibly find if there is a one-layer model also for the layerwise 2LSWE is to viscously regularize the equations by an added viscosity.
How viscosity looks in the shallow water framework is for instance given in.\citep{marche2007}
Adding viscosity smooths out discontinuities and renders the interaction term $h_1\nabla h_2$ well-defined.
The equations can then be investigated numerically by studying how shocks emerge when the viscosity coefficient is reduced.
They can also be investigated analytically by looking at travelling wave solutions inside the emerging shocks.

\section{Cases for the one-dimensional dam-break problem}
\label{sec:cases}
In this section we present the cases that will be used to investigate \cref{thm:oneLayerApproximation} numerically.
The cases represent variations of the one-dimensional \emph{dam-break problem}.\citep{leveque02}
In the \emph{two-layer} dam-break problem, a lighter fluid of height $h_1$ spreads on top of a heavier fluid of height $h_2$ as shown in \cref{fig:dambreak}.
The problem has been frequently used in the literature as a benchmark case for spreading models.\citep{joshi2018,soares-frazao2012,zhou2004}
\begin{figure}
  \centering
  \includegraphics[width=0.9\columnwidth]{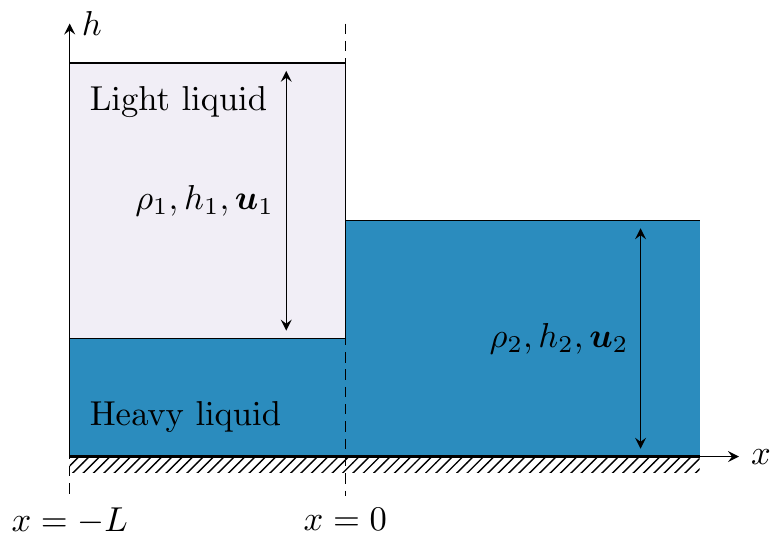}
  \caption{A simple sketch of the one-dimensional dam-break problem.}
  \label{fig:dambreak}
\end{figure}

In the following, we first consider the dam-break problem in an unrestricted spatial domain (``Case 0'').
This case will be used for convergence analyses.
We next consider the dam-break problem with a reflective wall boundary-condition (``Case R'' for ``reflective''), which is used both to compare qualitative differences between the forms of the 1LSWE and 2LSWE and to compare results with experimental data on two-layer spreading.
An overview of the cases is provided in \cref{fig:case-overview}.

\begin{figure}
  \centering
  \includegraphics[width=1.0\columnwidth]{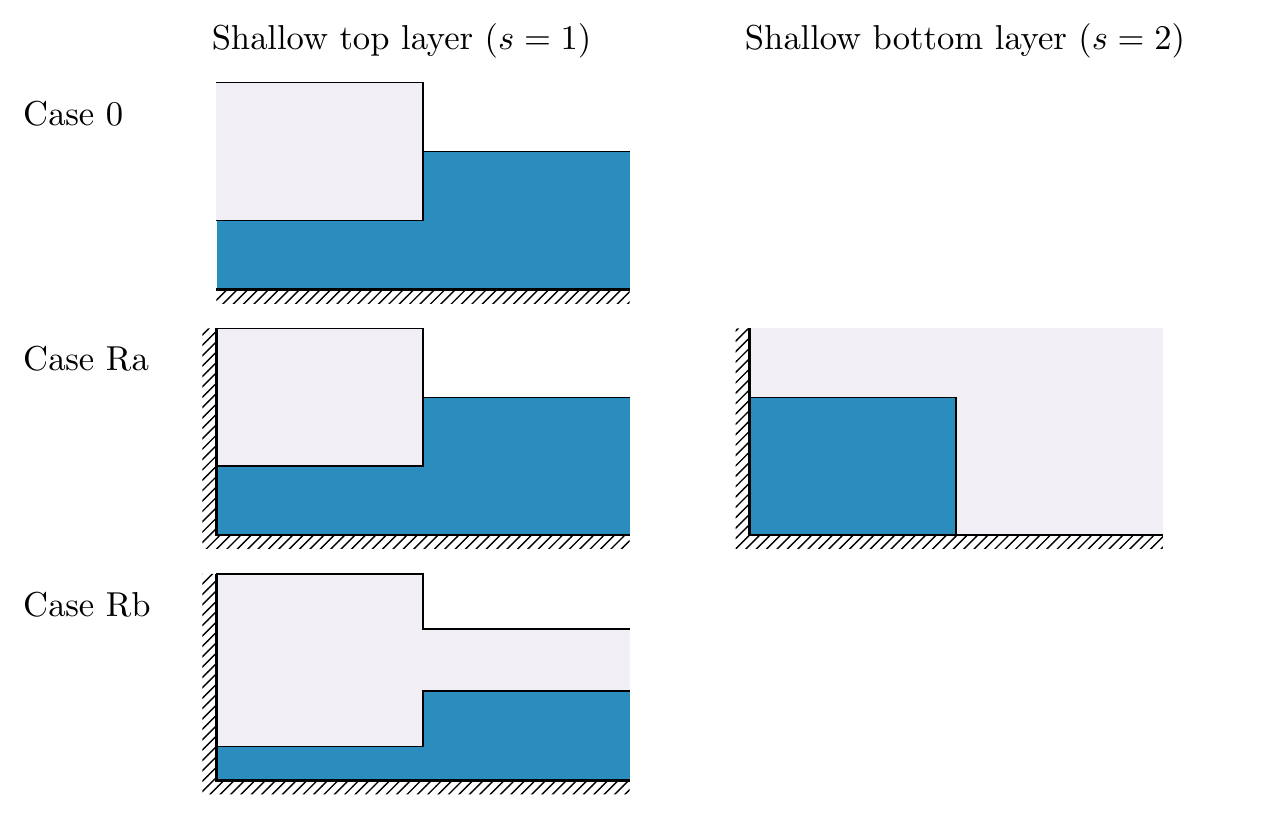}
  \caption{A tabular overview of initial conditions for the various test cases.
    The cases Ra, Rb and Rc all have reflecting walls to the left, and differ in the initial configuration of the fluids.
  Note that Case Rc uses the same initial conditions as Ra with $s=1$.}
  \label{fig:case-overview}
\end{figure}

\subsection{Case 0: Dam-break in an unrestricted spatial domain}
\label{sec:damBreak}
The initial conditions for the standard, one-dimensional dam-break problem that is not restricted in the flow-direction are
\begin{equation}
  \label{eq:dambreak}
  \begin{split}
    h_1(t=0, x) &=
    \begin{cases}
      h_0 & \text{if } x \le 0, \\
      0   & \text{if } x > 0,
    \end{cases} \\
    h_2(t=0, x) &= H - (1-\delta)h_1,\\
    u_1(t=0, x) &= 0, \\
    u_2(t=0, x) &= 0,
  \end{split}
\end{equation}
where $h_0$ is constant.

In this particular case, the corresponding one-layer problem has self-similar analytic solutions for both variants of the 1LSWE.
With the standard 1LSWE~\eqref{eq:1LSWEGlob}, there is the well-known Ritter solution,\citep{ritter1892}
\begin{subequations}
  \label{eq:ritterSol}
  \begin{align}
    h\left(x,t\right) &=
    \begin{cases}
      h_0 & \text{if } x \leq -c_0 t, \\
      \frac{h_0}{9}\left(2-\frac{x}{c_0 t}\right)^{2}
          & \text{if } -c_0 t < x \leq 2c_0 t, \\
      0   & \text{if } 2c_0 t < x,
    \end{cases}
    \label{eq:ritterSolH}\\
    u\left(x,t\right) &=
    \begin{cases}
      0 & \text{if } x \leq -c_0 t, \\
      \frac{2}{3}\left(c_0+\frac{x}{t}\right)
        & \text{if } -c_0 t < x \leq 2c_0 t, \\
      0 & \text{if } 2c_0 t < x,
    \end{cases}
  \end{align}
\end{subequations}
where $c_0 = \sqrt{\delta gh_0}$.
This solution is obtained from the assumption that \cref{eq:1LSWEGlob} is valid across discontinuities, as is normally the case when working with the 1LSWE.
For the locally conservative form~\eqref{eq:1LSWELoc}, the analytic solution is
\begin{widetext}
\begin{subequations}
  \label{eq:velSolDam}
  \begin{align}
    h\left(x,t\right) &=
    \begin{cases}
      h_0 & \text{if } x \leq -c_0 t, \\
      \frac{h_0}{9}\left(2-\frac{x}{c_0 t}\right)^{2}
    & \text{if } -c_0 t < x \leq \frac{(2-\sqrt{2})c_0 t}{1+\sqrt{2}}, \\
      \frac{4h_0}{(2+\sqrt{2})^2}
    & \text{if } \frac{(2-\sqrt{2})c_0 t}{1+\sqrt{2}} < x \leq \frac{2c_0 t}{1 +\sqrt{2}},
    \label{eq:velSolDamH}\\
    0   & \text{if } \frac{2c_0 t}{(1+\sqrt{2})} < x,
    \end{cases} \\
    u\left(x,t\right) &=
    \begin{cases}
      0 & \text{if } x \leq -c_0 t, \\
      \frac{2}{3}\left(c_0+\frac{x}{t}\right)
    & \text{if } -c_0 t < x \leq \frac{(2-\sqrt{2})c_0 t}{1+\sqrt{2}}, \\
      \frac{2c_0}{1+\sqrt{2}}
    & \text{if } \frac{(2-\sqrt{2})c_0 t}{1+\sqrt{2}} < x \leq \frac{2c_0 t}{1+\sqrt{2}}, \\
      0 & \text{if } \frac{2c_0 t}{1+\sqrt{2}} < x.
    \end{cases}
  \end{align}
\end{subequations}
\end{widetext}

A sketch of the two solutions for $h$ is shown in \cref{fig:ritter}.
One can see that the Ritter solution expands more than $2.4$ times faster than the solution of the locally conservative form.
The latter solution is the only one with a discontinuous height profile, and it has a constant Froude number of $\FrLe = \sqrt{2}$ at the leading edge.
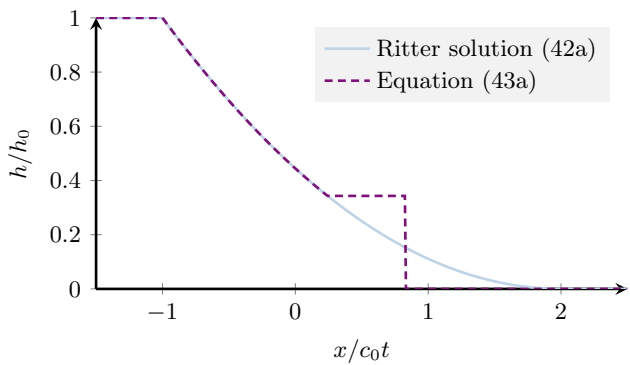
\begin{figure}
  \centering
  \tikzsetnextfilename{ritter}
  \begin{tikzpicture}
    \begin{axis}[
        axis lines = left,
        ylabel = $h/h_0$,
        xlabel = $x/c_0 t$,
        width = 1.0*\columnwidth,
        height = {0.6*\columnwidth},
        domain = -1.5:2.5,
        line width = 1pt,
        legend pos = north east,
        legend style = {draw=none, fill=black!5},
        legend cell align = {left},
        samples = 500,
        no marks,
      ]
      \addplot+[color=BuPu-D, solid]
        {x < -1
          ? 1
          : (x < 2
             ? (2 - x)^2 / 9
             : 0)};
      \addlegendentry{Ritter solution~\eqref{eq:ritterSolH}}

      \addplot+[color=BuPu-K, densely dashed]
        {x < -1
          ? 1
          : (x < ((2-sqrt(2))/(sqrt(2)+1))
             ? (2 - x)^2 / 9
             : (x < (2/(sqrt(2)+1))
                ? 4/(2+sqrt(2))^2
                : 0))};
      \addlegendentry{\Cref{eq:velSolDamH}}
    \end{axis}
  \end{tikzpicture}
  \caption{A sketch of the Ritter solution~\eqref{eq:ritterSol} and \cref{eq:velSolDam} for the one-layer dam-break problem.}
  \label{fig:ritter}
\end{figure}

\subsection{Case Ra: Quantify inaccuracies in the one-layer approximation}
\label{sec:experimentCaseRa}
In Case Ra, the initial conditions are the same as for Case 0 (\cref{eq:dambreak}).
However, a reflective wall is placed to the left of the dam at position $x=-L$ with boundary conditions $(\partial_x h)(x=-L, t) = 0$ and $u(x=-L, t) = 0$.
The reflective wall removes the self-similarity of the solution, which enables a study of how the accuracy of the one-layer approximation from \cref{thm:oneLayerApproximation} evolves in time.

We also consider a variant of this case where the top layer becomes deep, i.e. $s=2$ in \cref{thm:oneLayerApproximation}.
Here the initial conditions become
\begin{equation}
  \label{eq:dambreakBottom}
  \begin{split}
    h_1(t=0, x) &= H - h_2, \\
    h_2(t=0, x) &=
    \begin{cases}
      h_0 & \text{if } x \le 0, \\
      0   & \text{if } x > 0,
    \end{cases} \\
    u_1(t=0, x) &= 0, \\
    u_2(t=0, x) &= 0,
  \end{split}
\end{equation}
where again $h_0$ is constant.

\subsection{Case Rb: Effect of non-zero depth on both sides of dam}
\label{sec:experimentCaseRb}
Case Rb is a variant of case Ra where the initial conditions are relaxed to allow a non-zero depth to the right of the dam, that is,
\begin{equation}
  \label{eq:dambreakRb}
  \begin{split}
    h_1(t=0, x) &=
    \begin{cases}
      h_{0, a} & \text{if } x \le 0, \\
      h_{0, b} & \text{if } x > 0,
    \end{cases} \\
    h_2(t=0, x) &= H - (1-\delta)h_1,\\
    u_1(t=0, x) &= 0, \\
    u_2(t=0, x) &= 0.
  \end{split}
\end{equation}
In this case, the difference between the solutions of the locally and globally conservative 1LSWE will be less notable, because both give shocks.
This case will be used to show that the locally conservative 1LSWE captures quantitative behaviour of two-layer cases that is not captured by the globally conservative 1LSWE.

\subsection{Case Rc: Comparison to dam-break experiments}
\label{sec:experimentCases}
Finally, in case Rc we compare numerical results of the dam-break case with experimental results for liquid-on-liquid spreading.
In particular, we compare the spreading radius predicted by the one-layer approximation from \cref{thm:oneLayerApproximation} (\cref{eq:1LSWELocGen}) and by the two-layer equations to two sets of experimental results.
We use the same initial conditions as in case Ra, that is, \cref{eq:dambreak} with a reflective wall at $x=-L$.

\begin{table*}
  \caption{Initial conditions used by the one-dimensional dam-break experiments.
  The experiments by \citeauthor{suchon1970} are with oil spreading on water, while those from \citeauthor{chang1983} are liquified methane and liquified nitrogen spreading on water.}
  \label{tab:experiments}
  \begin{ruledtabular}
    \begin{tabular}{llD{.}{.}{2.1}D{.}{.}{2.2}D{.}{.}{2.1}D{.}{.}{3.3}}%
      Authors & Experiment
              & \multicolumn{1}{c}{Spill Volume (\si{\liter})}
              & \multicolumn{1}{c}{Height $h_0$ (\si{\centi\meter})}
              & \multicolumn{1}{c}{Width $L$ (\si{\centi\meter})}
              & \multicolumn{1}{c}{$\delta$} \\ \hline
      \citeauthor{suchon1970} & Run 11 &  10 & 16.51 & 10.16 & 0.1 \\
      \citeauthor{suchon1970} & Run 14 &  7.7 & 16.637 & 7.62 & 0.1 \\
      \citeauthor{suchon1970} & Run 17 &  5.1 & 10.9982 & 7.62 & 0.1 \\
      \citeauthor{suchon1970} & Run 18 &  5.1 & 16.51 & 5.08 & 0.1 \\
      \citeauthor{chang1983} & \SI{2}{\liter} methane &  2.0 & 17.3 & 7 & 0.746 \\
      \citeauthor{chang1983} & \SI{2}{\liter} nitrogen &  2.0 & 17.3 & 7 & 0.34 \\
      \citeauthor{chang1983} & \SI{0.75}{\liter} methane &  0.75 & 6.5 & 7 & 0.746 \\
      \citeauthor{chang1983} & \SI{1}{\liter} nitrogen &  1.0 & 8.7 & 7 & 0.34 \\
  \end{tabular}
  \end{ruledtabular}
\end{table*}
The first set of experiments is from \citeauthor{suchon1970},\cite{suchon1970} who studied the spreading of oil on water.
He used a \SI{2.5}{\meter} long and \SI{0.62}{\meter} wide channel with glass walls.
The initial dam was controlled by a thin aluminum plate that was manually removed to start the experiment.
We use initial conditions corresponding to 4 different runs by \citeauthor{suchon1970}, see \cref{tab:experiments}.
The initial depth of water in the experiments was about \SI{30}{\centi\meter}, which is nearly twice the initial heights.

The second set of experiments that will be considered are those presented by \citeauthor{chang1983}.\cite{chang1983}
They studied fluids at cryogenic temperatures (cryogens) spreading on water and presented both experimental results as well as model predictions.
In their model, they used the same empirical boundary condition for the spreading rate as discussed in \cref{sec:BenjaminAlternative}.\citep{chang1982}
We will demonstrate that their experimental results can be reproduced to a high accuracy without any empirical boundary condition or model for the spreading rate.

It should be noted that the experimental setup by \citet{chang1983} deviates from the dam-break case in that the initial reservoir of cryogen is emptied through a large slit.
The spreading then occurred inside a cylinder of length \SI{4}{\meter} with an inside diameter of \SI{16.5}{\centi\meter} where half of the volume was filled with water.
However, the case should be well approximated by a dam-break since the slit height is of the same order of magnitude as that of the leading edge of the spreading liquid.
To the best of our knowledge, the numerical predictions by \citeauthor{chang1983} were also based on the dam-break case.
\citeauthor{chang1983} do not list the initial height and width of the released cryogens, only the initial volumes.
The initial conditions are therefore estimated based on the description of the apparatus given by \citeauthor{chang1982}.\cite{chang1982}
We assume that the spreading occurs in a channel of the same width as that of the experiment.
We then estimated the area of the release tank and used this to find an estimate for the initial height and width from a given initial volume.
The initial conditions used are listed in \cref{tab:experiments}.

To accurately represent the spreading of cryogens, it is necessary to account for evaporation due to heat flow from the water and surrounding air.
The evaporation gives a source term in the mass conservation laws.
We follow \citet{chang1983} and include constant evaporation rates of \SI{0.16}{\kilogram\per\meter\squared\per\second} for methane and \SI{0.201}{\kilogram\per\meter\squared\per\second} for nitrogen in the mass balances.
These values are the same as those used by \citeauthor{chang1983}, which are based on experimental studies of the relevant substances.\citep{burgess1970}
Evaporation leads to the formation of bubbles in the liquid, which reduces its density.
\citeauthor{chang1983} called this reduced density the \emph{effective cryogenic density}, and they estimated it based on experimental results to be \SI{660}{\kilogram\per\meter\cubed} for nitrogen and \SI{254}{\kilogram\per\meter\cubed} for methane.
Similar to \citeauthor{chang1983}, we will use reduced densities in our simulations as well.
Finally, \citeauthor{chang1983} report the formation of some ice on the water surface downstream of the cryogen distributor.
This effect is not accounted for in the models, although it is also not expected to have a large impact on the spreading rates.

\section{Numerical results}
\label{sec:results}
In this section, we will discuss results from the cases described in \cref{sec:cases}.
The equations are discretized spatially using a finite-volume scheme. We employ the FORCE (first-order centered) flux~\citep{toro2000} and the second-order MUSCL (Monotonic Upstream-Centered Scheme for Conservation Laws) reconstruction with a minmod limiter~\citep{leveque02} in each finite volume.
The solutions are advanced in time with a standard third-order three-stage strong stability-preserving Runge-Kutte method.\citep{ketcheson05}
Although some of the equations have terms that are in general not conservative, \eg the convective term in the locally conservative 2LSWE~\eqref{eq:2LSWEGenVel}, it should be noted that these become conservative when the equations are restricted to a single spatial dimension.
The Courant–Friedrichs–Lewy (CFL)-number is \num{0.9} for all cases.

For quantification of errors, we use the $L^1$ norm, which for a function $y : \Omega \to \mathds{R}$ is defined as
\begin{equation}
  \norm{y} \equiv \int_{\Omega} \abs{y(\vect x)} \diff^n x.
\end{equation}

We present numerical results for the cases described in \cref{sec:cases} for $\delta$ of $0.6$ and $0.7$, which are similar to those of cryogenic spills on water (see for instance \cref{tab:experiments}).
To find the height and velocity distributions in the one-layer model for other values of $\delta$, a temporal scaling is all that is needed.
This is because the equations are invariant under the transformations $\delta g \to \lambda \delta g$, $t \to \sqrt{\lambda}^{-1} t$ and $u_s \to \sqrt{\lambda} u_s$ for all $\lambda$.
We find that the convergence is also similar for other values of $\delta$.

\subsection{Case 0: Dam-break in an unrestricted spatial domain}
In Case 0, the dam-break occurs in a one-dimensional, spatially unrestricted domain.
We solved this case with the parameters $h_0 = \SI{1}{\meter}$ and $\delta = 0.7$ with 400 grid cells on the domain $x/c_0 t\in(-2,2)$.
The initial depth, $H$, was increased stepwise from the initial value, $H = h_0$, to obtain a larger height difference between the layers.
The results presented in \Cref{fig:case0-qualitative} show that the locally conservative 2LSWE converge towards the analytic solution of the locally conservative 1LSWE when $H$ increases.
A higher value of $H$ translates into increasing $D_k$ in \cref{thm:oneLayerApproximation}.
The figure demonstrates a general trend found for the agreement between the 1LSWE and the 2LSWE, namely that the locally conserved 1LSWE become an increasingly good approximation to the complete 2LSWE with increasing $H$.
\begin{figure}
  \centering
  \tikzsetnextfilename{case0-qualitative}
  \begin{tikzpicture}
    \begin{axis}[
        height = 0.8\columnwidth,
        width = 1.0\columnwidth,
        xlabel = {$x/c_0 t$},
        ylabel = {$h_1/h_0$},
        axis lines = left,
        xmax = 2.1,
        ymin =-0.1,
        ymax = 0.8,
        line width = 1pt,
        legend pos = north east,
        legend style = {draw=none, fill=black!5},
        legend cell align = {left},
        hd1/.style = {
          densely dashdotted,
          color=Blues-E,
        },
        hd2/.style = {
          densely dotted,
          color=Blues-G,
        },
        hd3/.style = {
          densely dashed,
          color=Blues-J,
        },
        han/.style = {
          color=Set1-C,
        },
        initsol/.style = {
          thin,
          color=Set1-A,
        },
        table/x = x,
      ]

      \addplot[initsol] coordinates {
          (-2, 0.7)
          (0, 0.7)
          (0, 0)
          (2, 0)
        };
      \addlegendentry{Initial solution}

      \addplot[han] table [y=han] {data/case0_qualitative_local.txt};
      \addlegendentry{Analytic solution}

      \addplot[hd3] table [y=hd3] {data/case0_qualitative_local.txt};
      \addlegendentry{Depth $H=5h_0$}

      \addplot[hd2] table [y=hd2] {data/case0_qualitative_local.txt};
      \addlegendentry{Depth $H=2h_0$}

      \addplot[hd1] table [y=hd1] {data/case0_qualitative_local.txt};
      \addlegendentry{Depth $H=h_0$}
    \end{axis}
  \end{tikzpicture}
  \caption{Comparison of solutions of the local 2LSWE with increasing depths, $H$ (blue lines) to the analytic solution of the locally conservative 1LSWE for Case 0.
  As $H$ increases, the 2LSWE solution approaches the 1LSWE solution.}
  \label{fig:case0-qualitative}
\end{figure}
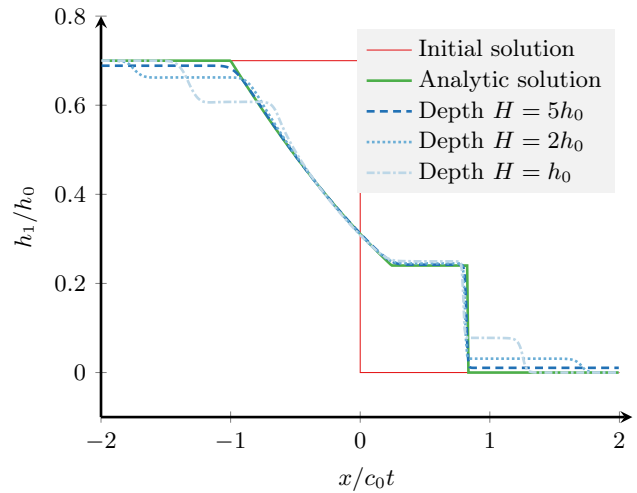

\Cref{fig:error} shows the normalized difference in $L^1$ for the top-layer height between the analytic solution~(\cref{eq:velSolDam}), and the solutions obtained with the 2LSWE, $\tilde h_1$ and $h_1$, respectively ($\lVert h_1 - \tilde h_1\rVert_1/\norm{h_0}$).
The differences are shown as a function of the initial depth, $H$.
The circles correspond to the locally conservative 2LSWE~\eqref{eq:2LSWEGenVel}, the triangles correspond to the globally conservative 2LSWE~\eqref{eq:2LSWEGenTot}, and the solid line indicates a slope of -1.
The plot shows that the difference between the globally and the locally conservative 2LSWE is small as expected.
Other relevant variables such as $h_2$, $u_1$, and $u_2$, were found to exhibit a similar behaviour.
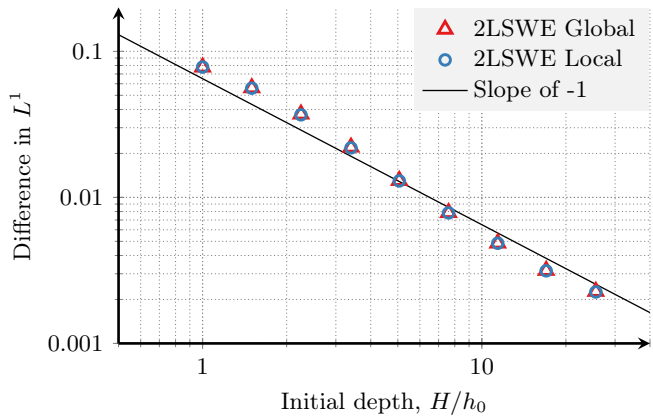
\begin{figure}
  \centering
  \tikzsetnextfilename{error}
  \begin{tikzpicture}
    \begin{axis}[
        height = 0.7*\columnwidth,
        width = 1.0*\columnwidth,
        xlabel = {Initial depth, $H/h_0$},
        ylabel = {Difference in $L^1$},
        xmode = log,
        ymode = log,
        ymin = 0.001,
        ymax = 0.2,
        xmin = 0.5,
        xmax = 40,
        axis lines = left,
        line width = 1pt,
        grid=both,
        grid style={densely dotted, line width=0.1pt, draw=black!50},
        major grid style={line width=0.2pt, draw=black!50},
        log ticks with fixed point,
        legend cell align = {left},
        legend style={
          draw=none,
          fill=black!5,
          at={(1,1)},
          anchor=north east,
        },
        locmom/.style={
          only marks,
          mark=o,
          mark size=2.0pt,
          mark options={Set1-B, line width=1.1pt},
        },
        totmom/.style={
          only marks,
          thick,
          mark=triangle,
          mark options={Set1-A, line width=1.1pt, scale=1.5},
        },
        orderone/.style={
          no marks,
          black,
          line width=0.5pt,
        },
        domain=0.5:40,
        samples=10,
        table/x expr = \thisrowno{0}+0.15,
        table/y expr = \thisrowno{1},
      ]

      \addplot+[totmom] table {data/L1diffTot.txt};
      \addlegendentry{2LSWE Global}
      \addplot+[locmom] table {data/L1diffLoc.txt};
      \addlegendentry{2LSWE Local}
      \addplot+[orderone] {0.065/x};
      \addlegendentry{Slope of -1}
    \end{axis}
  \end{tikzpicture}
  \caption{The $L^1$ difference of the top-layer height, $h_1$, between analytic solutions of Case 0 and solutions with the 2LSWE for different initial depths, $H$.
    The circles correspond to the locally conservative 2LSWE and the triangles correspond to the globally conservative 2LSWE.
  The line indicate a slope of -1.}
  \label{fig:error}
\end{figure}

\subsection{Case Ra: Quantify inaccuracies in the one-layer approximation}
In Case Ra, a reflective wall is placed at $x = -L$, cf.~\cref{sec:experimentCaseRa}.
The case was solved with the parameters $h_0 = \SI{1}{\meter}$, $L = \SI{2}{\meter}$, and $\delta = 0.6$.
The domain width was $\SI{15}{\meter}$ and 1000 grid cells were used.

We first compare the two situations $s=1$ (the top layer is shallow) and $s=2$ (the bottom layer is shallow), cf.~\cref{thm:oneLayerApproximation}.
\Cref{fig:heightDistribution} shows a dam break in the left column ($s=1$) and the cross section of a gravity current in the right column ($s=2$) at $t=\SI{3}{\second}$.
An illustration of the initial configurations is presented at the top of the figure.
The globally conservative 2LSWE (green lines) are compared with the locally conservative 1LSWE (green lines) and the globally conservative 1LSWE (red lines).
For the dam break case (right column), we initialize the bottom layer in a perturbed state where the depth $h_2$ is constant.
This is done to show that a perturbation of the initial solution of the relatively deep layer does not prevent the 2LSWE to converge to the one-layer approximation when the depth increases.
As the depths increase, we see that the solutions of the 2LSWE converge toward the locally conservative 1LSWE as predicted by the theorem.
\begin{figure*}
  \centering
  \tikzsetnextfilename{heightDistribution}
  \begin{tikzpicture}[
      singleloctop/.style={
        color=Paired-B,
        thin,
      },
      singlelocbot/.style={
        color=Paired-B,
        densely dashed,
      },
      singletottop/.style={
        color=Paired-E,
        thin,
      },
      singletotbot/.style={
        color=Paired-E,
        densely dashed,
      },
      doubletop/.style={
        color=Paired-D,
        thin,
      },
      doublebot/.style={
        color=Paired-D,
        densely dashed,
      },
      >=stealth,
    ]

    \begin{groupplot}[
        height=3.5cm,
        width=6.7cm,
        ylabel={Height (\si{\meter})},
        axis lines = left,
        table/x expr = \thisrow{x} - 2.0,
        group style={
          group size=2 by 4,
          vertical sep=0.5cm,
          xticklabels at=edge bottom,
        },
        bottomFlow/.style={
          ylabel = {},
          separate axis lines,
          y axis line style={draw=none},
        },
      ]

      \nextgroupplot[ymin=0.3]
      \addplot[singleloctop] table [y expr = 0.5+\thisrow{h1}] {data/distOneLayerLocMom.txt};
      \addplot[singlelocbot] table [y expr = 0.5-\thisrow{h2}] {data/distOneLayerLocMom.txt};
      \addplot[singletottop] table [y expr = 0.5+\thisrow{h1}] {data/distOneLayerTotMom.txt};
      \addplot[singletotbot] table [y expr = 0.5-\thisrow{h2}] {data/distOneLayerTotMom.txt};
      \addplot[doubletop] table [y expr = \thisrow{h1} + \thisrow{h2}] {data/distTwoLayerTotMom0.5m.txt};
      \addplot[doublebot] table [y expr = \thisrow{h2}] {data/distTwoLayerTotMom0.5m.txt};

      \nextgroupplot[bottomFlow, yticklabels={0.0,0.4,2.0}, ytick={0,0.4,1.0}]
      \addplot[singleloctop] table [y expr = 1.0] {data/distBottomOneLayerLocMom.txt};
      \addplot[singlelocbot] table [y expr = \thisrow{h1}] {data/distBottomOneLayerLocMom.txt};
      \addplot[singletottop] table [y expr = 1.0] {data/distBottomOneLayerTotMom.txt};
      \addplot[singletotbot] table [y expr = \thisrow{h1}] {data/distBottomOneLayerTotMom.txt};
      \addplot[doubletop] table [y expr = \thisrow{h1} + \thisrow{h2}-1.0] {data/distBottomTwoLayerTotMom2m.txt};
      \addplot[doublebot] table [y expr = \thisrow{h2}] {data/distBottomTwoLayerTotMom2m.txt};

      \nextgroupplot[ymin=0.8]
      \addplot[singleloctop] table [y expr = 1+\thisrow{h1}] {data/distOneLayerLocMom.txt};
      \addplot[singlelocbot] table [y expr = 1-\thisrow{h2}] {data/distOneLayerLocMom.txt};
      \addplot[singletottop] table [y expr = 1+\thisrow{h1}] {data/distOneLayerTotMom.txt};
      \addplot[singletotbot] table [y expr = 1-\thisrow{h2}] {data/distOneLayerTotMom.txt};
      \addplot[doubletop] table [y expr = \thisrow{h1} + \thisrow{h2}] {data/distTwoLayerTotMom1m.txt};
      \addplot[doublebot] table [y expr = \thisrow{h2}] {data/distTwoLayerTotMom1m.txt};

      \nextgroupplot[bottomFlow, yticklabels={0.0,0.4,4.0}, ytick={0,0.4,1.0}]
      \addplot[singleloctop] table [y expr = 1.0] {data/distBottomOneLayerLocMom.txt};
      \addplot[singlelocbot] table [y expr = \thisrow{h1}] {data/distBottomOneLayerLocMom.txt};
      \addplot[singletottop] table [y expr = 1.0] {data/distBottomOneLayerTotMom.txt};
      \addplot[singletotbot] table [y expr = \thisrow{h1}] {data/distBottomOneLayerTotMom.txt};
      \addplot[doubletop] table [y expr = \thisrow{h1} + \thisrow{h2}-3] {data/distBottomTwoLayerTotMom4m.txt};
      \addplot[doublebot] table [y expr = \thisrow{h2}] {data/distBottomTwoLayerTotMom4m.txt};

      \nextgroupplot[ymin=9.8]
      \addplot[singleloctop] table [y expr = 10+\thisrow{h1}] {data/distOneLayerLocMom.txt};
      \addplot[singlelocbot] table [y expr = 10-\thisrow{h2}] {data/distOneLayerLocMom.txt};
      \addplot[singletottop] table [y expr = 10+\thisrow{h1}] {data/distOneLayerTotMom.txt};
      \addplot[singletotbot] table [y expr = 10-\thisrow{h2}] {data/distOneLayerTotMom.txt};
      \addplot[doubletop] table [y expr = \thisrow{h1} + \thisrow{h2}] {data/distTwoLayerTotMom10m.txt};
      \addplot[doublebot] table [y expr = \thisrow{h2}] {data/distTwoLayerTotMom10m.txt};

      \nextgroupplot[bottomFlow, yticklabels={0.0,0.4,10}, ytick={0,0.4,1.0}]
      \addplot[singleloctop] table [y expr = 1] {data/distBottomOneLayerLocMom.txt};
      \addplot[singlelocbot] table [y expr = \thisrow{h1}] {data/distBottomOneLayerLocMom.txt};
      \addplot[singletottop] table [y expr = 1] {data/distBottomOneLayerTotMom.txt};
      \addplot[singletotbot] table [y expr = \thisrow{h1}] {data/distBottomOneLayerTotMom.txt};
      \addplot[doubletop] table [y expr = \thisrow{h1} + \thisrow{h2}-9] {data/distBottomTwoLayerTotMom10m.txt};
      \addplot[doublebot] table [y expr = \thisrow{h2}] {data/distBottomTwoLayerTotMom10m.txt};

      \nextgroupplot[ymin=49.8, xlabel={$x$ (\si{\meter})}]
      \addplot[singleloctop] table [y expr = 50+\thisrow{h1}] {data/distOneLayerLocMom.txt};
      \addplot[singlelocbot] table [y expr = 50-\thisrow{h2}] {data/distOneLayerLocMom.txt};
      \addplot[singletottop] table [y expr = 50+\thisrow{h1}] {data/distOneLayerTotMom.txt};
      \addplot[singletotbot] table [y expr = 50-\thisrow{h2}] {data/distOneLayerTotMom.txt};
      \addplot[doubletop] table [y expr = \thisrow{h1} + \thisrow{h2}] {data/distTwoLayerTotMom50m.txt};
      \addplot[doublebot] table [y expr = \thisrow{h2}] {data/distTwoLayerTotMom50m.txt};

      \nextgroupplot[bottomFlow, yticklabels={0.0,0.4,50}, ytick={0,0.4,1.0}, xlabel={$x$ (\si{\meter})}]
      \addplot[singleloctop] table [y expr = 1] {data/distBottomOneLayerLocMom.txt};
      \addplot[singlelocbot] table [y expr = \thisrow{h1}] {data/distBottomOneLayerLocMom.txt};
      \addplot[singletottop] table [y expr = 1] {data/distBottomOneLayerTotMom.txt};
      \addplot[singletotbot] table [y expr = \thisrow{h1}] {data/distBottomOneLayerTotMom.txt};
      \addplot[doubletop] table [y expr = \thisrow{h1} + \thisrow{h2}-49] {data/distBottomTwoLayerTotMom50m.txt};
      \addplot[doublebot] table [y expr = \thisrow{h2}] {data/distBottomTwoLayerTotMom50m.txt};
    \end{groupplot}

    \node[above right] (rb)
      at ([xshift=-0.1cm, yshift=0.3cm] group c1r1.north west)
      {\includegraphics[width=5.16cm, height=2.25cm]{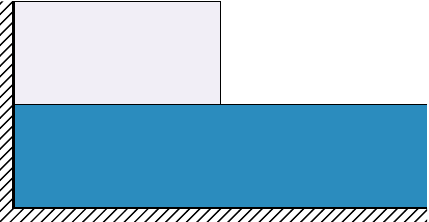}};
    \node[above right] (ra)
      at ([xshift=-0.1cm, yshift=0.3cm] group c2r1.north west)
      {\includegraphics[width=5.16cm, height=2.25cm]{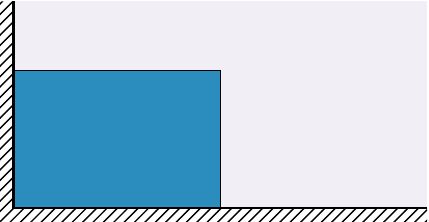}};
    \draw[thick, ->]
      ([xshift=-10ex] group c1r1.north west) --
      ([xshift=-10ex, yshift=-0.5cm] group c1r4.south west)
      node[midway, above, rotate=90] {Increasing $D_k$};

    \draw[thin, ->] (group c2r1.south west) -- (group c2r1.west) decorate[decoration=zigzag]{-- ($(group c2r1.west)!.5!(group c2r1.north west)$)} -- (group c2r1.north west);
    \draw[thin, ->] (group c2r2.south west) -- (group c2r2.west) decorate[decoration=zigzag]{-- ($(group c2r2.west)!.5!(group c2r2.north west)$)} -- (group c2r2.north west);
    \draw[thin, ->] (group c2r3.south west) -- (group c2r3.west) decorate[decoration=zigzag]{-- ($(group c2r3.west)!.5!(group c2r3.north west)$)} -- (group c2r3.north west);
    \draw[thin, ->] (group c2r4.south west) -- (group c2r4.west) decorate[decoration=zigzag]{-- ($(group c2r4.west)!.5!(group c2r4.north west)$)} -- (group c2r4.north west);
  \end{tikzpicture}
  \caption{Height distribution of the two layers in a dam-break problem at $t=\SI{3}{\second}$ solved with the locally conserved 1LSWE~\eqref{eq:1LSWELocGen} (blue lines), the globally conserved 1LSWE~\eqref{eq:1LSWEGlob} (red lines) and the globally conserved 2LSWE~\eqref{eq:2LSWEGenTot} (green lines).
  The solid lines and dashed lines indicate $h_1 + h_2$ and $h_2$, respectively.}
  \label{fig:heightDistribution}
\end{figure*}
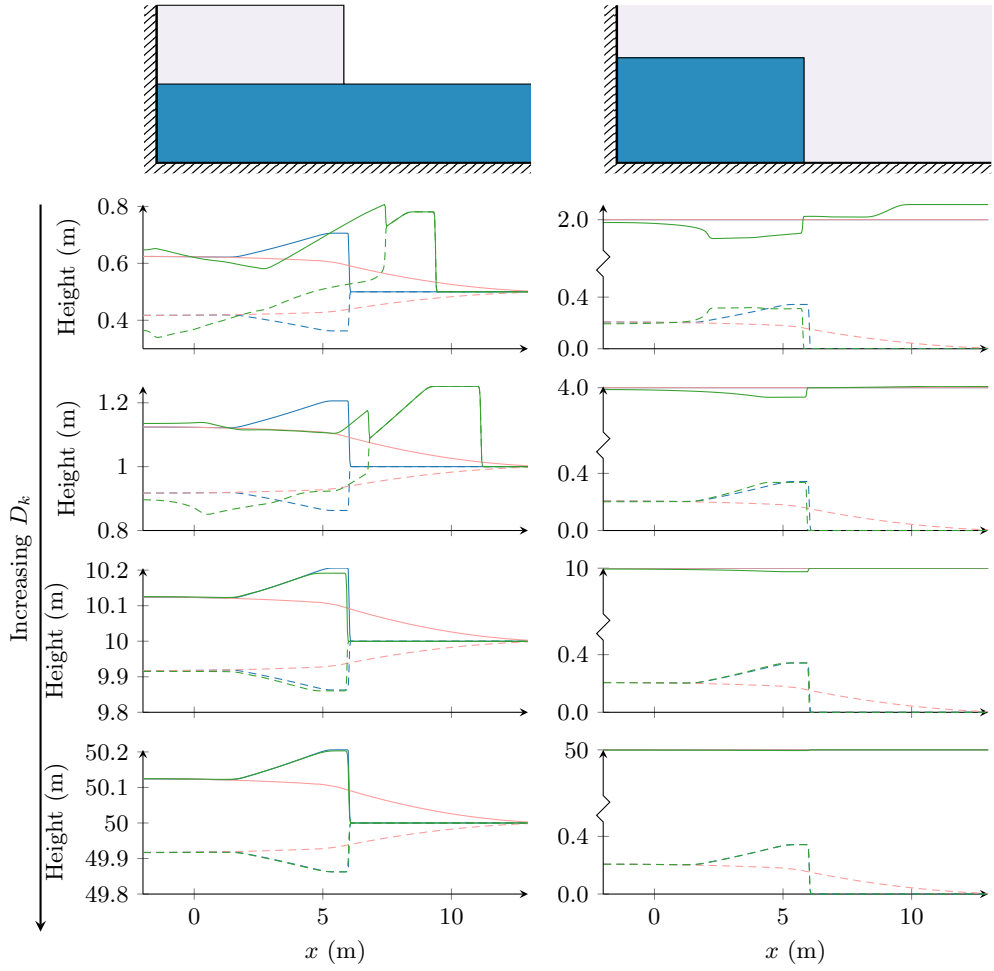

To quantify how the solutions of the 2LSWE converge to those of the locally conservative 1LSWE, we will compare the solutions at various times, initial depths, $H$, and initial widths, $L$.
We consider solutions of the globally conservative 2LSWE and the locally conservative 1LSWE and evaluate two quantifiable differences.
In \cref{fig:error3D}, the $L^1$ difference of the top-layer height, $\lVert h_1^\mathrm{2LSWE} - h_1^\mathrm{1LSWE}\rVert_1/h_0 L$ is plotted for varying times and depths, $H/h_0$.
\Cref{fig:errorSpreadingrate} shows the difference of the \emph{leading edge} position at $t = \SI{5}{\second}$,  $|r_\mathrm{2LSWE} - r_\mathrm{1LSWE}|/h_0$, for varying initial depths, $H/h_0$, and widths, $L/h_0$.
The position of the leading edge is here defined as the smallest $x$-value where the top layer is thinner than $\SI{e-4}{\meter}$.
\Cref{fig:caseRa_errors} shows that the differences in $h_1$ decrease with time, which is reasonable since the spreading fluid becomes gradually thinner.
As expected, the differences decrease with increasing value of $H/h_0$.
Similar to Case 0, the errors in the variables $h_2$, $u_1$, and $u_2$ as quantified by the $L^1$ norm exhibit the same trends as the top layer height (not shown).
Further, the figure shows that the difference of the leading edge position decreases with decreasing width, which is reasonable because the spreading fluid becomes thinner as the initial volume decreases.
\begin{figure}
  \centering
  \begin{subfigure}[t]{0.9\columnwidth}
    \centering
    \tikzsetnextfilename{error3D}
    \begin{tikzpicture}
      \begin{axis}[
          height = 1.0\columnwidth,
          width = 1.0\columnwidth,
          xlabel = {Initial depth, $H/h_0$},
          ylabel = {Time ($\si{\second}$)},
          xmax = 20,
          xmode = log,
          log ticks with fixed point,
          xtick={1, 2, 5, 10, 20},
          view = {0}{90},
          colormap = {violet}{rgb=(0.3,0.06,0.5), rgb=(0.81,0.81,0.765)},
          plotstyle/.style={
            contour gnuplot = {
              levels = {0.001, 0.005, 0.01, 0.02, 0.05, 0.1, 0.2, 0.3},
              label distance=80pt,
              contour label style={
                nodes={text=black},
                /pgf/number format/fixed,
                /pgf/number format/fixed zerofill=true,
                /pgf/number format/precision=3,
              },
            },
          },
        ]

        \addplot3[plotstyle] table {data/3Dh1DiffTot.txt};
      \end{axis}
    \end{tikzpicture}
    \caption{The $L^1$ difference of the top-layer height.}
    \label{fig:error3D}
  \end{subfigure} \\[1em]
  \begin{subfigure}[t]{0.9\columnwidth}
    \centering
    \tikzsetnextfilename{errorSpreadingrate}
    \begin{tikzpicture}
      \begin{axis}[
          height = 0.9*\columnwidth,
          width = 1.0*\columnwidth,
          xlabel = {Initial depth, $H/h_0$},
          ylabel = {Initial dam width, $L/h_0$},
          xmin = 1.0,
          xmax = 20.0,
          xmode = log,
          log ticks with fixed point,
          xtick={1, 2, 5, 10, 20},
          view = {0}{90},
          colormap = {violet}{rgb=(0.3,0.06,0.5), rgb=(0.9,0.9,0.85)},
          plotstyle/.style={
            contour gnuplot = {
              levels = {
                0.01, 0.03, 0.05, 0.1, 0.2, 0.4, 0.6, 0.8, 1.0, 1.2, 1.4, 1.5
              },
              contour label style={
                nodes={text=black},
                /pgf/number format/fixed,
                /pgf/number format/fixed zerofill=true,
                /pgf/number format/precision=2,
              },
            },
          },
        ]

        \addplot3[plotstyle] table {data/L1diffTotdepth-width.txt};
      \end{axis}
    \end{tikzpicture}
    \caption{The difference of the \emph{leading edge} position at $t = \SI{5}{\second}$.}
    \label{fig:errorSpreadingrate}
  \end{subfigure}
  \caption{A quantitative comparison of solutions from the globally conservative 2LSWE and the locally conservative 1LSWE for Case Ra, showing a) the difference in top-layer height, b) the difference in leading edge position.}
  \label{fig:caseRa_errors}
\end{figure}
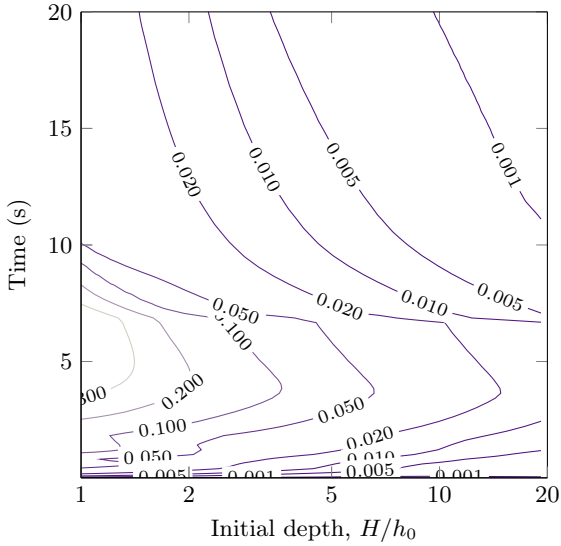
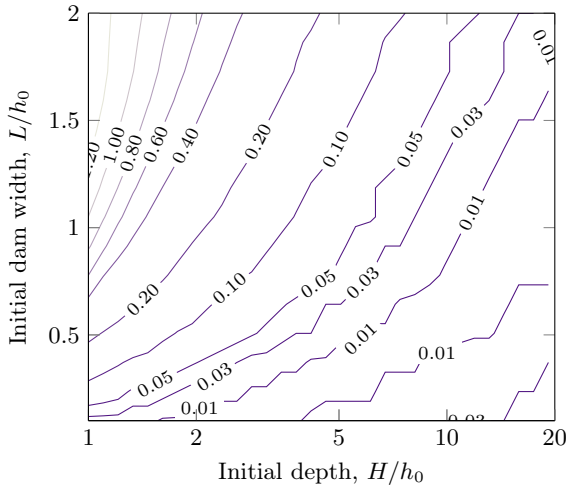

It is also interesting to see how the rate of spreading evolves with increasing depth, $H$.
\Cref{fig:spreadingRate} shows the leading edge position as a function of time for the two variants of the 1LSWE and the globally conservative 2LSWE for different depths $H$.
Again we observe a rapid convergence of the 2LSWE to the locally conservative 1LSWE.
Also for this case, we observe that the spreading rate from the globally conservative 1LSWE is higher that that from the full 2LSWE.
\begin{figure}
  \centering
  \tikzsetnextfilename{spreadingRate}
  \begin{tikzpicture}
  \begin{axis}[
        height = 0.9*\columnwidth,
        width = 1.0*\columnwidth,
        xlabel = {Time (\si{\second})},
        ylabel = {Leading edge (\si{\meter})},
        axis lines = left,
        xmax = 6.5,
        ymax = 15.0,
        line width = 1pt,
        legend cell align = {left},
        legend style={
          draw=none,
          fill=black!5,
          at={(0.55,0.07)},
          anchor=south west,
        },
        hd1/.style = {
          densely dashdotted,
          color=Blues-E,
        },
        hd3/.style = {
          densely dotted,
          line width=0.75pt,
          color=Blues-G,
        },
        hd9/.style = {
          densely dashed,
          color=Blues-J,
        },
        loc/.style = {
          color=Set1-C,
        },
        tot/.style = {
          dashed,
          color=Set1-A,
        },
        table/x = t,
        table/y = r,
      ]

      \addplot[loc] table  {data/leadingEdgeOneLayerLocMom.txt};
      \addplot[tot] table  {data/leadingEdgeOneLayerTotMom.txt};
      \addplot[hd9] table  {data/leadingEdgeTwoLayerTotMom9m.txt};
      \addplot[hd3] table  {data/leadingEdgeTwoLayerTotMom3m.txt};
      \addplot[hd1] table  {data/leadingEdgeTwoLayerTotMom1m.txt};

      \addlegendentry{1LSWE Local}
      \addlegendentry{1LSWE Global}
      \addlegendentry{Depth $H = 9h_0$}
      \addlegendentry{Depth $H = 3h_0$}
      \addlegendentry{Depth $H = h_0$}
    \end{axis}
  \end{tikzpicture}
  \caption{The position of leading edge as function of time for Case Ra. The blue lines shows the result for 2LSWE with various depths.}
  \label{fig:spreadingRate}
\end{figure}
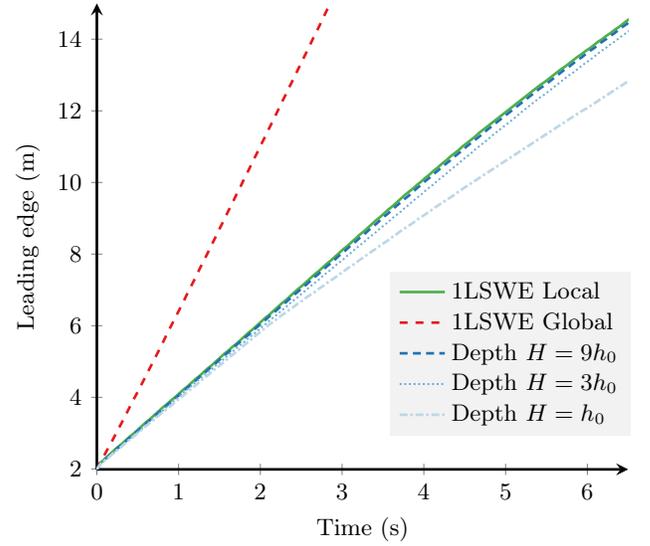

\subsection{Case Rb: Effect of non-zero depth at both sides of dam}
In Case Rb, the dam-break was initialized according to \cref{eq:dambreakRb} with a non-zero depth at both sides of the dam.
The case was solved with the parameters $\delta = 0.6$, $L=\SI{4}{\meter}$, $H=\SI{50}{\meter}$, $h_{0,a} = \SI{2}{\meter}$, and $h_{0,b} = \SI{0.5}{\meter}$.
The width of the domain is \SI{15}{\meter} and the results are again computed with 1000 grid cells.

\Cref{fig:Riemann} shows the height distributions at time $t=\SI{3}{\second}$ for both the globally and locally conserved 1LSWE and for both versions of the 2LSWE, \cref{eq:2LSWEGenVel,eq:2LSWEGenTot} at different depths $H$.
The solutions of both formulations of the 1LSWE have shocks, but the shock velocities differ.
As expected, the height profiles of the locally and globally conservative 2LSWE are similar, however, they are only in agreement with the locally conservative 1LSWE.
The figure shows that the locally conservative 1LSWE should be used for accurate representation of the position of the leading edge.
\begin{figure}
  \centering
  \tikzsetnextfilename{Riemann}
  \begin{tikzpicture}
    \begin{axis}[
        height=4.5cm,
        width=1.0*\columnwidth,
        ylabel={Height $(\si{\meter})$},
        xlabel={$x$ (\si{\meter})},
        ylabel style={align=center},
        xmin=-4,
        xmax=12.5,
        ymin=49.4,
        ymax=50.9,
        axis lines=left,
        legend cell align=left,
        legend columns=2,
        legend style={
          draw=none,
          fill=black!5,
          at={(0.1,1.3)},
          anchor=north west,
          /tikz/column 2/.style={column sep=5pt,},
        },
        singleloc/.style={
          color=Paired-E,
          solid,
        },
        singletot/.style={
          color=Paired-F,
          densely dotted,
        },
        doubleloc/.style={
          color=Paired-C,
          dashed,
        },
        doubletot/.style={
          color=Paired-D,
          densely dotted,
        },
        no markers,
        line width = 1pt,
        table/x expr = \thisrow{x} - 4,
        table/y expr = \thisrow{h1} + \thisrow{h2},
      ]

      \addplot+[singleloc, forget plot]
        table [y expr = 50+\thisrow{h1}] {data/distRiemannOneLayerLocMom.txt}
        node[black, below] {$h_1 + h_2$};
      \addplot+[singleloc]
        table [y expr = 50-\thisrow{h2}] {data/distRiemannOneLayerLocMom.txt}
        node[black, below] {$h_2$};
      \addlegendentry{1LSWE Loc}

      \addplot+[singletot, forget plot]
        table [y expr = 50+\thisrow{h1}] {data/distRiemannOneLayerTotMom.txt};
      \addplot+[singletot]
        table [y expr = 50-\thisrow{h2}] {data/distRiemannOneLayerTotMom.txt};
      \addlegendentry{1LSWE Glob}

      \addplot+[doubleloc, forget plot]
        table {data/distRiemannTwoLayerLocMom.txt};
      \addplot+[doubleloc]
        table [y expr = \thisrow{h2}] {data/distRiemannTwoLayerLocMom.txt};
      \addlegendentry{2LSWE Loc}

      \addplot+[doubletot, forget plot]
        table {data/distRiemannTwoLayerTotMom.txt};
      \addplot+[doubletot]
        table [y expr = \thisrow{h2}] {data/distRiemannTwoLayerTotMom.txt};
      \addlegendentry{2LSWE Glob}
    \end{axis}
  \end{tikzpicture}
  \caption{Height profiles of $h_2$ and $h_1 + h_2$ at $t=\SI{3}{\second}$ of solutions to Case Rb with the different formulations of the 1LSWE and 2LSWE.
  ``Loc'' and ``Glob'' denote the locally conservative and globally conservative formulation, respectively.}
  \label{fig:Riemann}
\end{figure}
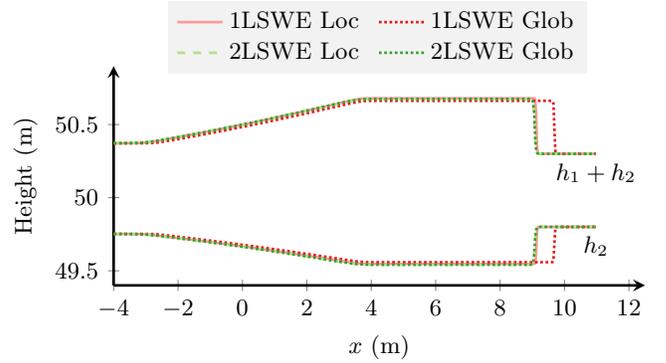

\subsection{Case Rc: Comparison to spreading experiments}
\label{sec:experiments}
\Cref{fig:suchon,fig:chang} present a comparison of the spreading rates predicted from the one-layer approximation from \cref{thm:oneLayerApproximation} (green solid line) with the experiments described in \cref{sec:experimentCases} (symbols) and with the full 2LSWE (blue dashed lines).
All cases were solved with 2000 grid cells.
A comparison to the simpler Fay model (\cref{eq:fay}) with $\beta = 1.31$ is also included for the Suchon experiments (orange dotted lines).

We find good agreement between the one-layer approximation and available experimental data.
The deviation in the spreading radius was calculated as the relative difference in the $L^1$ norm, that is,
\[
  \mathrm{dev}(r_\mathrm{exp}, r_\mathrm{sim})
  = \frac{\|r_\mathrm{exp} - r_\mathrm{sim}\|}{\|r_\mathrm{exp}\|}.
\]
The deviation in the spreading radius is \SI{4.5}{\percent} for oil on water, which is significantly better than the Fay model, where the average deviations are \SI{12.6}{\percent}.
For the cryogenic fluids, the deviations in the spreading radius were \SI{10.2}{\percent} for methane on water, and \SI{4.2}{\percent} for nitrogen on water.

We remark that in the experiments, the depth of the water is not much larger than the initial depth of the spreading liquid.
In the experiments by \citeauthor{suchon1970}, it is about twice the initial height of the oil, and in the cryogen experiments by \citeauthor{chang1983}, it is about the same as the initial cryogen height.
However, it was shown in \cref{fig:errorSpreadingrate} that the difference between the predicted spreading distance from the 2LSWE and the one-layer approximation after $\SI{5}{\second}$ is still small, even at these initial depths.
In all cases, the initial width $L$ is less than half the initial width, and so we expect a deviation at $\SI{5}{\second}$ that is smaller than 0.2 times the initial height $h_0$.

A comparison of the green solid lines (1LSWE) and the blue dashed lines (2LSWE) in \cref{fig:suchon} confirms a very good agreement between the two formulations.
For \cref{fig:chang}, and especially for the \SI{2}{\liter} cases where the initial height is large compared to the water depth, the discrepancy is larger.
We find that the deviation in the $L^1$ norm between the one-layer approximation and 2LSWE results is between \SI{0.7}{\percent} and \SI{5.5}{\percent} for all cases.
In the \SI{2}{\liter} cryogen experiments, we observe that the discrepancy is reducing after some time.
This is consistent with \cref{fig:error3D}.
The evaporation that occurs during the spreading of cryogenic fluids likely accelerates the decrease in error.

The results indicate that the proposed one-layer approximation may be used as an approximation to the 2LSWE even for cases where the depth ratio is small, as long as the main interest is to predict the spreading distance.
Although, in these cases one should not expect that the one-layer approximation captures all of the qualitative flow patterns that are captured by the 2LSWE.
In \cref{fig:chang-profiles}, we compare the profiles of the top layer at three different times for the 1LSWE (green solid lines) and the 2LSWE (blue dashed lines).
The observed spreading distance from the corresponding experiments are marked by red dots.
We see that the 2LSWE captures a more complex behaviour, especially in the early phase of the flow, but as expected, the agreement between the profiles improves with time.

Finally, we note that \citet{chang1983} also solved the 1LSWE numerically, but with an imposed boundary condition with $\FrLe = 1.28$ at the leading edge.
They motivated the use of a constant Froude number at the leading edge by frequent use in previous literature dealing with spreading of non-boiling fluids.
They treated the value of $\FrLe$ as a parameter that depends on the apparatus and must be determined experimentally and found that $\FrLe = 1.28$ worked best for their apparatus.
The analysis in this paper shows that the constant Froude number can in fact be derived from the 2LSWE.
That is, as the relative depth of the bottom layer increases, $\FrLe$ approaches $\sqrt{2}$ from below.
Moreover, our analysis shows that $\FrLe = \sqrt{2}$ is true also for boiling liquids and even when including other relevant source terms in the 2LSWE model.

\begin{figure}
  \centering
  \tikzsetnextfilename{suchon}
  \begin{tikzpicture}[
      label/.style={
        fill=black!5,
        above left=1ex,
        thin,
        draw,
      },
    ]

    \begin{groupplot}[
        width = 0.55*\columnwidth,
        height = 0.55*\columnwidth,
        xlabel = {Time (\si{\second})},
        ylabel = {Length (\si{\centi\meter})},
        xmin = 0,
        xmax = 11,
        ymin = 0,
        ymax = 220,
        group style={
          group size= 2 by 2,
          vertical sep=0.5cm,
          horizontal sep=0.5cm,
          xticklabels at=edge bottom,
          yticklabels at=edge left
        },
        grid=both,
        grid style={line width=.1pt, draw=gray!10},
        major grid style={line width=.2pt,draw=gray!50},
        legend style={line width=1.0pt},
        exp/.style={
          only marks,
          mark=*,
          mark size=1.5pt,
          mark options={Set1-A},
          table/skip first n=1,
        },
        sm1/.style={
          mark=none,
          color=Set1-C,
        },
        sm2/.style={
          mark=none,
          densely dashed,
          color=Set1-B,
        },
        fay/.style={
          mark=none,
          semithick,
          densely dotted,
          color=Set1-D,
        },
        table/x expr = \thisrow{t},
        table/y expr = \thisrow{r}*100,
        domain = 0:11,
        samples = 200,
      ]

      \newcommand*{\FayPlot}[2]{
        100*(1.5*1.31*sqrt(0.1*9.81*#1*#2)*x)^(0.666666667)
      }

      \nextgroupplot[xlabel={}]
      \addplot+[exp] table[y expr = \thisrow{r}*100 - 10.16] {data/suchon11Exp.txt};
      \addplot+[sm1] table {data/suchon11Sim_1LSWE.txt};
      \addplot+[sm2] table {data/suchon11Sim_2LSWE.txt};
      \addplot+[fay] {\FayPlot{0.1016}{0.1651}};

      \nextgroupplot[xlabel={}, ylabel={}]
      \addplot+[exp] table[y expr = \thisrow{r}*100 - 7.62] {data/suchon14Exp.txt};
      \addplot+[sm1] table {data/suchon14Sim_1LSWE.txt};
      \addplot+[sm2] table {data/suchon14Sim_2LSWE.txt};
      \addplot+[fay] {\FayPlot{0.0762}{0.16637}};

      \nextgroupplot
      \addplot+[exp] table[y expr = \thisrow{r}*100 - 7.62] {data/suchon17Exp.txt};
      \addplot+[sm1] table {data/suchon17Sim_1LSWE.txt};
      \addplot+[sm2] table {data/suchon17Sim_2LSWE.txt};
      \addplot+[fay] {\FayPlot{0.109982}{0.0762}};

      \nextgroupplot[ylabel={}]
      \addplot+[exp] table[y expr = \thisrow{r}*100 - 5.08] {data/suchon18Exp.txt};
      \addplot+[sm1] table {data/suchon18Sim_1LSWE.txt};
      \addplot+[sm2] table {data/suchon18Sim_2LSWE.txt};
      \addplot+[fay] {\FayPlot{0.1651}{0.0508}};
    \end{groupplot}

    \node[label] at (group c1r1.south east) {Run number 11};
    \node[label] at (group c2r1.south east) {Run number 14};
    \node[label] at (group c1r2.south east) {Run number 17};
    \node[label] at (group c2r2.south east) {Run number 18};

    \matrix[
      anchor=south east,
      inner sep=0.5ex,
      row sep=0.25ex,
      column sep=1ex,
      nodes={anchor=west, text width=21ex, inner sep=0pt},
      fill=black!5,
    ] at ([yshift=1.5ex] group c2r1.north east) {
      \draw[Set1-C] (0,0) -- (0.5,0);
        & \node {\strut Simulation 1LSWE};
      \draw[densely dashed, Set1-B] (3,0) -- (3.5,0);
        & \node {\strut Simulation 2LSWE}; \\
      \draw[semithick, densely dotted, Set1-D] (0,0) -- (0.5,0);
        & \node {\strut Fay model};
      \fill[Set1-A] (3.25,0) circle(1.5pt);
        & \node {\strut Experiment}; \\
    };
  \end{tikzpicture}
  \caption{The spreading distance of oil on water as function of time.
  A comparison of results from the locally conservative 1LSWE, the globally conservative 2LSWE, the Fay model, and experimental data from \citeauthor{suchon1970}.\cite{suchon1970}}
  \label{fig:suchon}
\end{figure}
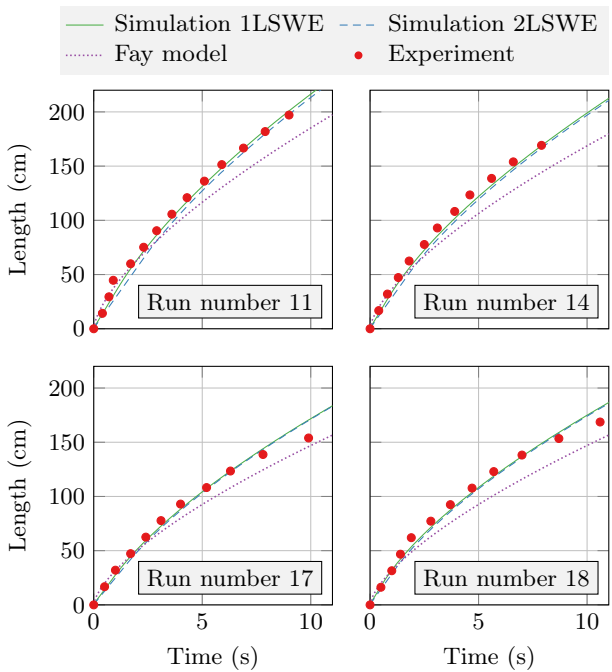

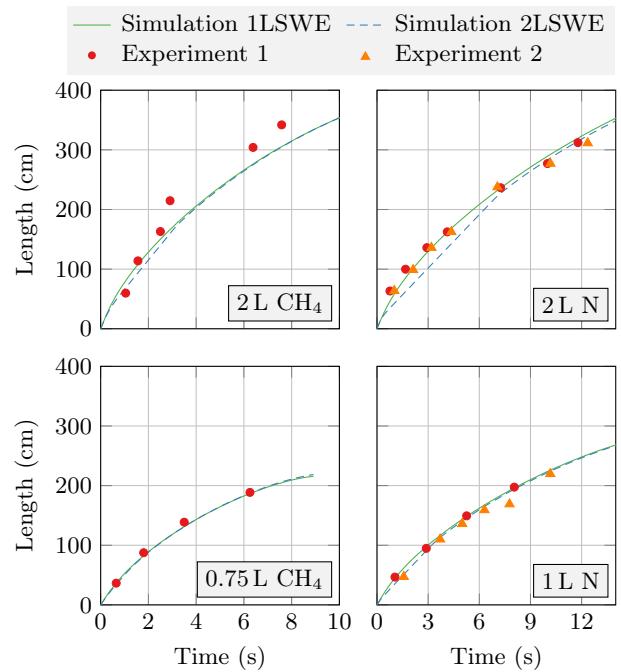
\begin{figure}
  \centering
  \tikzsetnextfilename{chang}
  \begin{tikzpicture}[
      label/.style={
        fill=black!5,
        above left=1ex,
        thin,
        draw,
      },
    ]

    \begin{groupplot}[
        width = 0.55*\columnwidth,
        height = 0.55*\columnwidth,
        xlabel = {Time (\si{\second})},
        ylabel = {Length (\si{\centi\meter})},
        xmin = 0,
        xmax = 14,
        ymin = 0,
        ymax = 400,
        group style={
          group size= 2 by 2,
          vertical sep=0.5cm,
          horizontal sep=0.5cm,
          xticklabels at=edge bottom,
          yticklabels at=edge left
        },
        grid=both,
        grid style={line width=.1pt, draw=gray!10},
        major grid style={line width=.2pt,draw=gray!50},
        ex1/.style={
          only marks,
          mark=*,
          mark size=1.5pt,
          mark options={Set1-A},
          table/y expr = \thisrow{r}*100,
        },
        ex2/.style={
          only marks,
          mark=triangle*,
          mark size=2pt,
          mark options={Set1-E},
          table/y expr = \thisrow{r}*100,
        },
        sim/.style={
          mark=none,
          color=Set1-C,
        },
        sm2/.style={
          mark=none,
          densely dashed,
          color=Set1-B,
        },
        table/x expr = \thisrow{t},
        table/y expr = \thisrow{r}*100,
        domain = 0:11,
        samples = 200,
      ]

      \nextgroupplot[xlabel={}, xmax=10]
      \addplot[ex1] table {data/chang_CH4_2L_Exp.txt};
      \addplot[sim] table {data/chang_CH4_2L_1LSWE.txt};
      \addplot[sm2] table {data/chang_CH4_2L_2LSWE.txt};

      \nextgroupplot[xlabel={}, ylabel={}, xtick distance=3]
      \addplot[ex1] table {data/chang_N2_2L_Exp_1.txt};
      \addplot[ex2] table {data/chang_N2_2L_Exp_2.txt};
      \addplot[sim] table {data/chang_N2_2L_1LSWE.txt};
      \addplot[sm2] table {data/chang_N2_2L_2LSWE.txt};

      \nextgroupplot[xmax=10]
      \addplot[ex1] table {data/chang_CH4_075L_Exp.txt};
      \addplot[sim] table {data/chang_CH4_075L_1LSWE.txt};
      \addplot[sm2] table {data/chang_CH4_075L_2LSWE.txt};

      \nextgroupplot[ylabel={}, xtick distance=3]
      \addplot[ex1] table {data/chang_N2_1L_Exp_1.txt};
      \addplot[ex2] table {data/chang_N2_1L_Exp_2.txt};
      \addplot[sim] table {data/chang_N2_1L_1LSWE.txt};
      \addplot[sm2] table {data/chang_N2_1L_2LSWE.txt};
    \end{groupplot}

    \node[label] at (group c1r1.south east) {\SI{2}{\liter} CH$_4$};
    \node[label] at (group c2r1.south east) {\SI{2}{\liter} N};
    \node[label] at (group c1r2.south east) {\SI{0.75}{\liter} CH$_4$};
    \node[label] at (group c2r2.south east) {\SI{1}{\liter} N};

    \matrix[
      anchor=south east,
      inner sep=0.5ex,
      row sep=0.25ex,
      column sep=1ex,
      nodes={anchor=west, text width=21ex, inner sep=0pt},
      fill=black!5,
    ] at ([yshift=1.5ex] group c2r1.north east) {
      \draw[Set1-C] (0,0) -- (0.5,0);
        & \node {\strut Simulation 1LSWE};
      \draw[densely dashed, Set1-B] (3,0) -- (3.5,0);
        & \node {\strut Simulation 2LSWE}; \\
      \fill[Set1-A] (0.25,0) circle(1.5pt);
        & \node {\strut Experiment 1};
      \node[
        anchor=center,
        isosceles triangle,
        isosceles triangle apex angle=60,
        rotate=90,
        fill=Set1-E,
        text width=0.12cm,
      ] at (3.25,0) {};
        & \node {\strut Experiment 2}; \\
    };
  \end{tikzpicture}
  \caption{The spreading distance of liquid methane (CH$_4$) and nitrogen (N) on water as function of time.
  A comparison of results from the locally conservative 1LSWE and experiments by \citeauthor{chang1983}.\cite{chang1983}}
  \label{fig:chang}
\end{figure}

\begin{figure*}
  \centering
  \tikzsetnextfilename{chang-profiles}
  \begin{tikzpicture}[
      timelabel/.style={black, fill=white},
    ]

    \begin{axis}[
        height=7cm,
        width=0.8*\textwidth,
        ylabel={Height (\si{\centi\meter})},
        xlabel={Length (\si{\centi\meter})},
        ylabel style={align=center},
        xmin=-7,
        xmax=360,
        ymin=-2.1,
        ymax=2.1,
        axis lines=left,
        legend cell align=left,
        legend style={
          draw=none,
          fill=black!5,
          at={(1.05,1)},
          anchor=north east,
        },
        1lswe/.style={
          color=Set1-C,
          thick,
        },
        1lsweforget/.style={
          1lswe,
          forget plot,
        },
        2lswe/.style={
          color=Set1-B,
          densely dashed,
          thick,
        },
        2lsweforget/.style={
          2lswe,
          forget plot,
        },
        expr/.style={
          only marks,
          mark=*,
          mark size=2.0pt,
          mark options={Set1-A},
        },
        table/y expr = 100*\thisrow{h1},
        table/x expr = 100*\thisrow{x},
      ]

      \addplot[1lswe] table[
        restrict x to domain=-7:81.5,
        ] {data/chang_CH4_2L_1LSWE_t=1.047.txt};
      \addlegendentry{One-layer approximation}
      \addplot[1lsweforget] table[
          restrict x to domain=-7:81.5,
          y expr=-100*\thisrow{h2},
        ] {data/chang_CH4_2L_1LSWE_t=1.047.txt};
      \addplot[2lswe] table[restrict x to domain=-7:69.5]
        {data/chang_CH4_2L_2LSWE_t=1.047.txt};
      \addlegendentry{Globally conservative 2LSWE}
      \addplot[2lsweforget] table[
          restrict x to domain=-7:69.5,
          y expr=100*\thisrow{h2},
        ] {data/chang_CH4_2L_2LSWE_t=1.047.txt};

      \addplot[1lsweforget] table[
          restrict x to domain=-7:150.5,
        ] {data/chang_CH4_2L_1LSWE_t=2.503.txt};
      \addplot[1lsweforget] table[
          restrict x to domain=-7:150.5,
          y expr=-100*\thisrow{h2},
        ] {data/chang_CH4_2L_1LSWE_t=2.503.txt};
      \addplot[2lsweforget] table[
          restrict x to domain=-7:140.5,
        ] {data/chang_CH4_2L_2LSWE_t=2.503.txt};
      \addplot[2lsweforget] table[
          restrict x to domain=-7:140.5,
          y expr=100*\thisrow{h2},
        ] {data/chang_CH4_2L_2LSWE_t=2.503.txt};

      \addplot[1lsweforget] table[
          restrict x to domain=-7:276.2,
        ] {data/chang_CH4_2L_1LSWE_t=6.384.txt};
      \addplot[1lsweforget] table[
          restrict x to domain=-7:276.2,
          y expr=-100*\thisrow{h2},
        ] {data/chang_CH4_2L_1LSWE_t=6.384.txt};
      \addplot[2lsweforget] table[
          restrict x to domain=-7:275.5,
        ] {data/chang_CH4_2L_2LSWE_t=6.384.txt};
      \addplot[2lsweforget] table[
          restrict x to domain=-7:275.5,
          y expr=100*\thisrow{h2},
        ] {data/chang_CH4_2L_2LSWE_t=6.384.txt};

      \draw[densely dotted] (axis cs: -7, 0) -- (axis cs: 360, 0);

      \addplot[expr] coordinates {
        (59.75, 0)
        (163.1, 0)
        (304.1, 0)
      };
      \addlegendentry{Experiments (spreading distance)}

      \node[timelabel] at (axis cs: 59.75, -1.05) {$t = \SI{1.05}{\second}$};
      \node[timelabel] at (axis cs: 163.1, -1.05) {$t = \SI{2.50}{\second}$};
      \node[timelabel] at (axis cs: 304.1, -1.05) {$t = \SI{6.38}{\second}$};
    \end{axis}
  \end{tikzpicture}
  \caption{A comparison of height profiles produced by the one-layer approximation and the globally conservative 2LSWE at different times for the \SI{2}{\liter} liquid methane case.}
  \label{fig:chang-profiles}
\end{figure*}
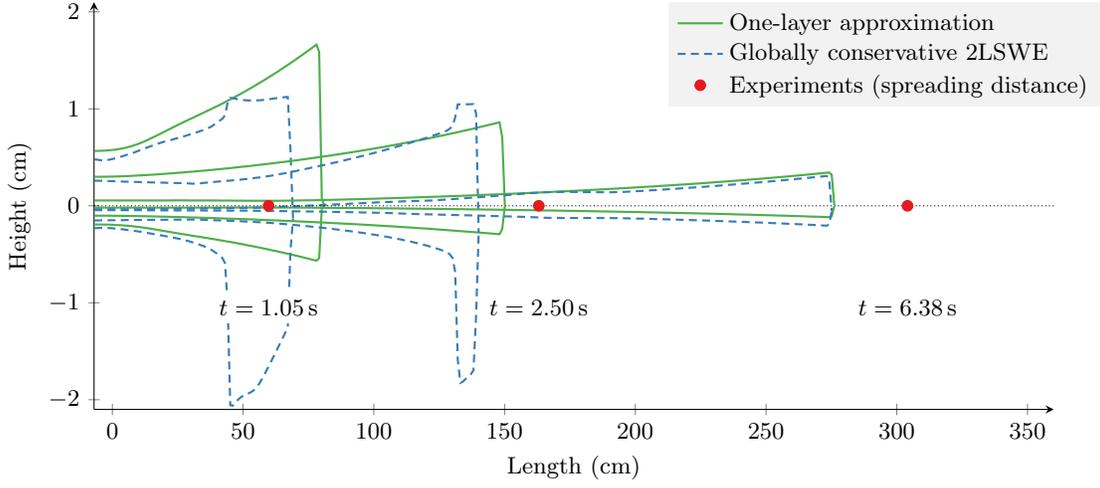

\section{Conclusions}
\label{sec:conclusion}
We have presented a comprehensive study of two-layer spreading where the depth of one layer is significantly larger than the other.
The main result is that the two-layer shallow-water equations can be approximated by an effective one-layer model with an effective gravitational constant as described in \cref{thm:oneLayerApproximation}.
In the literature, the globally conservative one-layer momentum equations are frequently used.
We have demonstrated both analytically and numerically that the locally conservative momentum equations should be used instead for a precise representation with an effective one-layer model.

Earlier works in the literature have made use of an additional boundary condition for the speed of the leading edge as a closure relation for the one-layer spreading model.
The speed is typically represented in terms of a constant Froude number which is adjusted to match experimental data so that $\FrLe \in [1,\sqrt{2}]$.
We have shown that this boundary condition can in fact be derived from the full two-layer shallow water equations.
By using the locally conservative version of the one-layer shallow water equations with the effective gravitational constant $(1-\rho_1/\rho_2)g$, the one-layer model correctly captures the behaviour of shocks and contact discontinuities.
In particular, the one-layer model results in $\FrLe = \sqrt{2}$, which is exactly the same as the theoretical predictions by \citet{karman1940,benjamin1968,ungarish2017} in the limit captured by \cref{thm:oneLayerApproximation}.

By using the same mathematical tools that were used to derive the one-layer approximation in \cref{thm:oneLayerApproximation}, we derived an expression for the Froude number at the front of a spreading fluid inside a rectangular cavity from the full two-layer shallow water equations.
The expression that we obtained from the analysis of the shallow-water equations is in good agreement with the expression by \citeauthor{benjamin1968}.
The agreement between these expressions suggests that the validity breakdown of the shallow-water equations in vicinity of shocks is less severe than previously suggested.

We compared to available experimental data for one-dimensional dam break experiments and found good agreement between the one-layer model derived in this work and experiments, where the mean relative deviation in the spreading radius was \SI{4.5}{\percent} for oil on water, \SI{10.2}{\percent} for methane on water, and \SI{4.2}{\percent} for nitrogen on water.
The spreading radius from the one-layer and two-layer descriptions could hardly be distinguished from each other after 10 seconds of spreading, but the fluid profiles from the two formulations differed at short times.
In comparison, the mean relative deviation in the spreading radius of the Fay model was \SI{12.6}{\percent} for oil on water.

The treatment in this paper has also included source terms, as long as they are source-bounded.
Source terms representing Coriolis forces are not source-bounded as defined in \cref{sec:reducingToOneLayer} because they are proportional to the depth.
Thus they are not covered by the present analysis.
It should be possible to include Coriolis-like source terms in the analysis, because although they are proportional to the depth, they are also proportional to the flow velocity which vanishes with increasing depth in the deep layer.
Since Coriolis forces are relevant particularly for the modelling of geophysical phenomena, it represents an attractive possibility for future work.

\section{Acknowledgements}
The authors wish to acknowledge fruitful discussions with Hans Langva Skarsvåg and Svend Tollak Munkejord.
This work was undertaken as part of the research project ``Predicting the risk of rapid phase-transition events in LNG spills (Predict-RPT)'', and the authors would like to acknowledge the financial support of the Research Council of Norway under the MAROFF programme (Grant 244076/O80).

\appendix
\section{Deriving Rankine-Hugoniot Conditions}
\label{app:travellingWave}
To obtain Rankine-Hugoniot condition for the 2LSWE we use an approach similar to that of \citeauthor{smoller1983}.\cite{smoller1983}
Consider a shock along $\Gamma$ which is normal to $\n$ and let $\phi$ be a test function with compact support $D$ which lies in the $x_nt$-plane.

Consider first the locally conservative system.
Because the velocity $\vect u_i$ is a weak solution, the normal component must satisfy
\begin{multline}
  \int_{D}\Biggl(\phi_t \ndot u_i + \phi_n \Biggl[\frac 1 2 (\ndot u_i)^2
  + \left(\frac{\rho_1}{\rho_2}\right)^{i-1}gh_1 \\
  + gh_2 + gb\Biggr] - \phi J_i\Biggr)\diff x_n\diff t = 0
  \label{eq:velWeakForm}
\end{multline}
where the subscript on $\phi$ denote partial differentiations and
\begin{equation}
  J_i = u_i^T\partial_T \left(\ndot u_i\right) - \frac{\n \cdot (\vect G_{h_i u_i} - \vect u_i G_{h_i})}{\rho_i h_i},
\end{equation}
$u_i^T$ is the tangential component of $\vect u_i$ and $\partial_T$ is differentiation with respect to the tangential direction.
The integrand in \cref{eq:velWeakForm} is a normal function and hence we can seperate the integral into two, one for each region where the velocities and heights are differentiable.
Call these regions $D_1$ and $D_2$.
Because the solution is differentiable inside these regions we may use Green's theorem and obtain
\begin{widetext}
\begin{multline}
  \int_{D_j}\left(\phi_t \ndot u_i + \phi_n \left[\frac 1 2 (\ndot u_i)^2 + \left(\frac{\rho_1}{\rho_2}\right)^{i-1}gh_1 + gh_2 + gb\right] - \phi J_i\right)\diff x_n\diff t \\
  =
  \int_{D_j}\left(\pd t (\phi \ndot u_i) + \pd {x_n}\left[\phi \left(\frac 1 2 \left(\ndot u_i\right)^2 + \left(\frac{\rho_1}{\rho_2}\right)^{i-1}gh_1 + gh_2 + gb\right)\right]\right)\diff x_n\diff t \\
  = \pm \lim_{\varepsilon\to 0}\int_{\Gamma \pm \varepsilon x_n}\phi \left(\left[\frac 1 2 \left(\ndot u_i\right)^2 + \left(\frac{\rho_1}{\rho_2}\right)^{i-1}gh_1 + gh_2 + gb\right]\diff t - \ndot u_i \diff x_n\right),
  \label{eq:velWeakPartj}
\end{multline}
\end{widetext}
because $\phi$ vanish on the boundary of $D$.
\Cref{eq:velWeakForm} is obtained by adding \cref{eq:velWeakPartj} with $j = 1$ and $j = 2$.
Because \cref{eq:velWeakForm} must hold for all test functions we obtain the Rankine-Hugoniot condition for the normal velocity component,
\begin{equation}
  S\jump{\ndot u_i} = \jump{\frac 1 2 (\ndot u_i)^2 + g \left(\frac{\rho_1}{\rho_2}\right)^{i-1}h_1 + gh_2}.
\end{equation}
For the globally conservative system a similar treatment yields
\begin{multline}
  S\jump{\rho_1h_1\vect u_1 + \rho_2h_2\vect u_2} \\
  = \jump{(\ndot u_1)\rho_1h_1\vect u_1 + (\ndot u_2)\rho_2h_2\vect u_2} \\
  + \jump{\frac 1 2 g \rho_1 h_1^2  + \rho_1 g h_1 h_2 + \frac 1 2 \rho_2 g h_2^2}\n
\end{multline}
and
\begin{multline}
  S\nd\jump{\vect u_2 - \vect u_1} =
  \Biggl\llbracket
    \frac 1 2\biggl[(\ndot u_2)^2 \\
    - (\ndot u_1)^2\biggr] - g\delta h_1
  \Biggr\rrbracket.
\end{multline}

The treatment presented above is not applicable for the tangential component of the velocity equations because they involve a term on the form $\ndot u_i\partial_n u_i^T$.
Note, that the tangential velocity components does not enter any of the other Rankine-Hugoniot conditions and that the equations are consistent if the tangential component are continuous across shocks.
Requiring $\jump{u_i^T} = 0$ has the additional advantage of making $\ndot u_i\partial_n u_i^T$ well defined as the product of  $\ndot u_i$ and $\partial_n u_i^T$.
Otherwise the distribution $\ndot u_i\partial_n u_i^T$ can not be decomposed without relying on some mollification scheme.

\citet{ostapenko2001} proposed that in the two-dimensional case one can use a Rankine-Hugoniot condition for the vorticity instead of the tangential velocity component.
Unfortunately, the proposed equation works only if one assumes $\jump{u_i^T} = 0$.
A conservation law for the quantity $\partial_1 u_{i,2} - \partial_2 u_{i,1} \equiv w_i$ can be obtained by taking distributional derivatives of the different components of the velocity equation, yielding
\begin{equation}
  \pdfun{w_i}{t} + \div \left(w_i\vect u_i + \vect J_i^{\perp}\right)  = 0,
  \label{eq:rotCon}
\end{equation}
where
\begin{equation}
  \vect J_i^{\perp} =
  \frac 1 {\rho_i h_i}
  \begin{pmatrix}
    G_{h_iu_i,2} - u_{i,2}G_{h_i} \\
    -G_{h_iu_i,1} + u_{i,1}G_{h_i}
  \end{pmatrix}.
\end{equation}
It may be tempting from \cref{eq:rotCon} to conclude that the vorticity must obey the jump condition
\begin{equation}
  S\jump{w_i} = \jump{w_i\ndot u_i + \ndot J_i^{\perp}}.
\end{equation}
However, this is only true if the vorticity $w_i$ can be interpreted as a normal function.
If the tangential velocity component is discontinuous across the shock, the vorticity would have a contribution similar to a delta distribution at the shock.
In that case one can not seperate the integral into two as was done in the derivation above, and the Rankine-Hugoniot condition would gain an additional contribution from the delta-like term.

For the purposes of this paper we can ignore the tangential velocity components across jumps.
The relevant equation in the one layer system is equal in both the globally conservative two-layer system and in the locally conservative two-layer system in the relevant limits.
Solutions of the two-layer systems are therefore also solutions of the locally conservative one-layer system.

\section{Rankine-Hugoniot conditions for 2LSWE}
\label{app:2LSWE-Rankine-Hugoniot}
In this appendix, we will show that the Rankine-Hugoniot conditions for the 2LSWE may be written as \cref{eq:2LSWECombinedRH}, repeated here for convenience,
\begin{subequations}
  \begin{align}
    \label{eq:CombinedRHa}
    S\jump{\rho_s h_s} &= \nd\jump{\rho_s h_s\vect u_s}, \\
    S\nd\jump{\vect u_s} &= \jump{\frac 1 2 (\ndot u_s)^2 + \delta g h_s}
    \nonumber \\
    &+ g_1(\gamma, S, h_s, \ndot u_s, \ndot u_d), \\
    \jump{\rho_1^{d-1}h_1 + \rho_2^{d-1}h_2} &= g_2(\gamma, S, h_s, \ndot u_s, \ndot u_d), \\
    S\nd \jump{\vect u_d} &= g_3(\gamma, S, h_s, \ndot u_s, \ndot u_d).
  \end{align}
  \label{eq:CombinedRH}
\end{subequations}
Here $g_1$ and $g_2$ differ for \cref{eq:2LSWEGenVel} and \cref{eq:2LSWEGenTot}, while $g_3$ will be the same.
Further, we will show that all of $g_1, g_2$, and $g_3$ vanish when $\gamma = 0$.
Note that \cref{eq:CombinedRHa} follows directly from \cref{eq:rankineHugoniotScalar} applied to the shallowest layer.

We first consider $g_3$, which can be obtained from mass conservation of layer $d$.
The scalar Rankine-Hugoniot condition~\eqref{eq:rankineHugoniotScalar} immediately yields
\begin{multline}
  S\jump{h_d}
    = \nd\jump{h_d \vect u_d} \nonumber \\
    \implies \n\cdot\jump{\vect u_d}
    = \frac{\jump{h_d}}{\av{h_d}}\left(S - \av{\ndot u_d}\right)
    = \frac{g_3}{S},
  \label{eq:g3}
\end{multline}
where we used that $\jump{ab} = \jump a \av b + \av a \jump b$.
This gives
\begin{equation}
  g_3 = \gamma S \jump{h_d} \left(S - \av{\ndot u_d}\right).
  \label{eq:g3}
\end{equation}

Next we consider the expressions for $g_1$ and $g_2$.
We first consider the locally conservative momentum equations~(\eqref{eq:2LSWEGenVel_u1} and~\eqref{eq:2LSWEGenVel_u2}).
We apply the Rankine-Hugoniot condition~\eqref{eq:2LSWEGenVelRH} with $i=d$ and insert \cref{eq:g3} to obtain
\begin{equation}
  \jump{\rho_1^{d-1}h_1 + \rho_2^{d-1}h_2} = \frac{\rho_2^{d-1}\jump{h_d}}{g\av{h_d}}(S - \av{\ndot u_d})^2,
  \label{eq:udRan}
\end{equation}
that is,
\begin{equation}
  g_2 = \frac{\rho_2^{d-1}}{g}\gamma \jump{h_d}(S - \av{\ndot u_d})^2.
  \label{eq:g2Loc}
\end{equation}

Now consider \cref{eq:2LSWEGenVelRH} with $i=s$,
\begin{equation}
  S\nd\jump{\vect u_s} = \jump{\frac 1 2 (\ndot u_s)^2
  + g\left(\frac{\rho_1}{\rho_2}\right)^{s-1}h_1 + gh_2}.
\end{equation}
If we consider the cases $s=1$ and $s=2$ separately and use that $d - 1 = 2 - s$, we find from \cref{eq:udRan} that
\begin{multline}
  \jump{g\left(\frac{\rho_1}{\rho_2}\right)^{s-1}h_1 + gh_2}
  = \jump{\delta g h_s} \\
  + \left(\frac{\rho_1}{\rho_2}\right)^{s-1}\gamma \jump{h_d}(S - \av{\ndot u_d})^2,
\end{multline}
where $\delta = 1 - \rho_1/\rho_2$, as defined in \cref{eq:delta}.
Inserting into \cref{eq:g2Loc} we get that
\begin{equation}
  S\nd\jump{\vect u_s} = \jump{\frac 1 2 (\ndot u_s)^2 + \delta g h_s} + g_1
\end{equation}
with
\begin{equation}
  g_1 = \left(\frac{\rho_1}{\rho_2}\right)^{s-1}\gamma \jump{h_d}(S - \av{\ndot u_d})^2.
  \label{eq:g1Loc}
\end{equation}

Finally, we consider the globally conservative momentum equations~(\eqref{eq:2LSWEGenTot_sumMom} and~\eqref{eq:2LSWEGenTot_diffVel}).
We first consider the case where $d=2$.
If we take the scalar product of the Rankine-Hugoniot conditions \cref{eq:2LSWEGenTotRH_diffVel,eq:2LSWEGenTotRH_totMom} with $\n$ and use \cref{eq:g3}, we obtain
\begin{widetext}
  \begin{subequations}
    \begin{align}
      \jump{\rho_1h_1 + \rho_2h_2} &=
        \frac 1 {g\av{h_2}}\Biggl(
          S\jump{\rho_1h_1\ndot u_1}
          - \jump{\rho_1h_1(\ndot u_1)^2}
          + S\rho_2 \jump{h_2}\av{\ndot u_2}
          - \rho_2\jump{h_2}\av{(\ndot u_2)^2}
          \nonumber \\
          &\hspace{3cm}- \jump{\frac 1 2 g \rho_1 h_1^2}
          - \rho_1g\av{h_1}\jump{h_2}
          + \rho_2\jump{h_2}(S - \av{\ndot u_2})
      (S - 2 \av{\ndot u_2})\Biggr), \\
      S\nd\jump{\vect u_1} &= \jump{\frac 1 2 (\ndot u_1)^2
                           + g\delta h_1}
                           + \frac{\jump{h_2}}{\av{h_2}}(S - \av{\ndot u_2})^2,
    \end{align}
  \end{subequations}
\end{widetext}
that is,
\begin{equation}
  g_1 = \gamma \jump{h_2}(S - \av{\ndot u_2})^2
\end{equation}
and
\begin{multline}
  g_2 = \frac{\gamma}{g}\Biggl(
      S\jump{\rho_1h_1\ndot u_1}
      - \jump{\rho_1h_1(\ndot u_1)^2} \\
      + S\rho_2 \jump{h_2}\av{\ndot u_2}
      - \rho_2\jump{h_2}\av{(\ndot u_2)^2} \\
      - \jump{\frac 1 2 g \rho_1 h_1^2}
      - \rho_1g\av{h_1}\jump{h_2} \\
    + \rho_2\jump{h_2}(S - \av{\ndot u_2})(S - 2 \av{\ndot u_2})\Biggr).
\end{multline}

Next we consider the case when $d=1$.
We may then write \cref{eq:2LSWEGenTotRH_totMom} as
\begin{multline}
  \jump{h_1 + h_2}
  = g_2 = \frac{\gamma}{g}\Biggl(
    \frac{S}{\rho_1}\jump{\rho_2h_2\ndot u_2} \\
    - \frac{1}{\rho_1}\jump{\rho_2h_2(\ndot u_2)^2}
    + \jump{h_1}\Bigl(S\av{\ndot u_1}\\ - \av{(\ndot u_1)^2}\Bigr)
    - g\av{h_2}\jump{h_1}
    - \frac{\rho_2}{2\rho_1}\jump{h_2^2} \\
    + \jump{h_1}\left(S-\av{\ndot u_1}\right)\left(S - 2\av{\ndot u_1}\right)\Biggr),
    \label{eq:g2GlobThm2}
\end{multline}
We then insert this into \cref{eq:2LSWEGenTotRH_diffVel} to get
\begin{multline}
  S\nd\jump{\vect u_2} = \jump{\frac 1 2 (\ndot u_2)^2 + \delta g h_2} \\
  + \left(\gamma \jump{h_1}(S - \av{\ndot u_1})^2 +\delta g\, g_2\right),
\end{multline}
which gives
\begin{equation}
  g_1 = \left(\gamma\jump{h_1}(S - \av{\ndot u_1})^2 +\delta g\, g_2\right).
\end{equation}

%

\end{document}